\documentclass[journal]{IEEEtran}
\usepackage{cite}
\usepackage{lscape}
\usepackage{acronym}
\usepackage{amsfonts,amssymb,amsmath,bbm,mathtools,stackrel,bbm,epsfig}
\usepackage{amsthm}
\usepackage{enumerate}
\usepackage{transparent}
\usepackage{graphicx,import}
\usepackage{algorithm,algorithmic}       
\usepackage{subcaption}
\usepackage{xr} 
\usepackage[all=normal,paragraphs=tight,floats=tight,mathspacing=normal,wordspacing=tight,charwidths=tight,mathdisplays=normal,leading=normal]{savetrees}

\usepackage[rgb]{xcolor} 
\usepackage[bookmarks,colorlinks]{hyperref}

\newcommand{\avec}{{\bf{a}}}

\newcommand{\bvec}{{\bf{b}}}

\newcommand{\dvec}{{\bf{d}}}
\newcommand{\evec}{{\bf{e}}}

\newcommand{\yvec}{{\bf{y}}}
\newcommand{\uvec}{{\bf{u}}}
\newcommand{\wvec}{{\bf{w}}}
\newcommand{\xvec}{{\bf{x}}}
\newcommand{\zvec}{{\bf{z}}}

\newcommand{\rvec}{{\bf{r}}}

\newcommand{\vvec}{{\bf{v}}}
\newcommand{\gvec}{{\bf{g}}}

\newcommand{\onevec}{{\bf{1}}}
\newcommand{\zerovec}{{\bf{0}}}

\newcommand{\Amat}{{\bf{A}}}
\newcommand{\Bmat}{{\bf{B}}}

\newcommand{\Dmat}{{\bf{D}}}
\newcommand{\Emat}{{\bf{E}}}

\newcommand{\Hmat}{{\bf{H}}}

\newcommand{\Imat}{{\bf{I}}}
\newcommand{\Lmat}{{\bf{L}}}
\newcommand{\Kmat}{{\bf{K}}}
\newcommand{\Pmat}{{\bf{P}}}

\newcommand{\Qmat}{{\bf{Q}}}

\newcommand{\Smat}{{\bf{S}}}

\newcommand{\Rmat}{{\bf{R}}}
\newcommand{\Umat}{{\bf{U}}}
\newcommand{\Vmat}{{\bf{V}}}
\newcommand{\Wmat}{{\bf{W}}}
\newcommand{\Xmat}{{\bf{X}}}
\newcommand{\Ymat}{{\bf{Y}}}

\newcommand{\define}{\stackrel{\triangle}{=}}

\newcommand{\be}{\begin{equation}}
\newcommand{\ee}{\end{equation}}
\newcommand{\beqna}{\begin{eqnarray}}
\newcommand{\eeqna}{\end{eqnarray}}

\acrodef{admm}[ADMM]{alternating
direction method of multiplier}
\acrodef{ac}[AC]{alternating current}
\acrodef{arma}[ARMA]{auto-regressive moving average}
\acrodef{bcrb}[bCRB]{biased CRB}
\acrodef{bmse}[BMSE]{Bayesian MSE}
\acrodef{cfar}[CFAR]{constant false alarm rate}
\acrodef{crb}[CRB]{ Cram$\acute{\text{e}}$r-Rao  bound\acroextra{-}}
\acrodef{cs}[cs]{compressed sensing}
\acrodef{dau}[DAU]{deep algorithm unrolling}
\acrodef{dc}[DC]{direct current}
\acrodef{dnn}[DNN]{deep neural network}
\acrodef{dsp}[DSP]{digital signal processing}
\acrodef{dl}[DL]{Deep Learning}
\acrodef{ems}[EMS]{energy management system}
\acrodef{ekf}[EKF]{extended Kalman filter}
\acrodef{evd}[EVD]{eigenvalue decomposition}
\acrodef{fim}[FIM]{Fisher information matrix}
\acrodef{gans}[GANs]{generative
adversarial networks}
\acrodef{gso}[GSO]{Graph shift operator}
\acrodef{gsp}[GSP]{graph signal processing}
\acrodef{gnn}[GNNs]{graph neural networks}
\acrodef{glrt}[GLRT]{generalized likelihood ratio test}
\acrodef{gso}[GSO]{graph shift operator}
\acrodef{gft}[GFT]{graph Fourier transform}
\acrodef{gfr}[GFR]{graph filter regularized}
\acrodef{gmrf}[GMRF]{Gaussian Markov random field}
\acrodef{gcs}[GCS]{graph convolutional skip}
\acrodef{gls}[GLS]{generalized least squares}
\acrodef{gcn}[GCNs]{graph convolutional networks}
\acrodef{gwkf}[GW-KF]{graph weights Kalman filter}
\acrodef{gf}[GF]{Graph-Filtered}
\acrodef{gd}[GD]{gradient descent}
\acrodef{gtkf}[GT-KF]{graph topology Kalman filter}
\acrodefplural{gtkf}[GT-KFs]{graph topology Kalman filters}
\acrodef{hpf}[graph HPF]{high-pass graph filter}
\acrodefplural{hpf}[graph HPFs]{high-pass graph filters}
\acrodef{iid}[i.i.d.]{independent and identically distributed}
\acrodef{ista}[ISTA]{iterative soft thresholding algorithm}
\acrodef{kkt}[KKT]{Karush-Kuhn-Tucker}
\acrodef{kf}[KF]{Kalman Filter}
\acrodef{lrt}[LRT]{likelihood ratio test}
\acrodef{lista}[LISTA]{learned \ac{ista}}
\acrodef{lmmse}[LMMSE]{linear minimum mean-squared error}
\acrodef{lpf}[graph LPF]{low-pass graph filter}
\acrodefplural{lpf}[graph LPFs]{low-pass graph filters}
\acrodef{map}[MAP]{maximum a-posteriori probability}
\acrodef{ml}[ML]{maximum likelihood}
\acrodef{mse}[MSE]{mean-squared-error}
\acrodef{mmse}[MMSE]{minimum MSE}
\acrodef{nns}[NNs]{neural networks}
\acrodef{pdf}[PDF]{probability density function}
\acrodef{pf}[PF]{power flow}
\acrodef{pfa}[PFA]{probability of false alarm}
\acrodef{psd}[PSD]{positive semi-definite}
\acrodef{psse}[PSSE]{power system state estimation}
\acrodef{pmu}[PMUs]{phasor measurement units}
\acrodef{pc}[PC]{point cloud}
\acrodef{pca}[PCA]{principal component analysis}
\acrodef{pgd}[PGD]{projected gradient descent}
\acrodef{roc}[ROC]{Receiver Operating Characteristic}
\acrodef{snr}[SNR]{signal-to-noise ratio}
\acrodef{svd}[SVD]{singular value decomposition}
\acrodef{ssm}[SSM]{state-space model}
\acrodef{tv}[TV]{total variation}
\acrodef{wcmse}[WC-MSE]{worst-case MSE}
\acrodef{wcbmse}[WC-BMSE]{worst-case BMSE}
\acrodef{wls}[WLS]{weighted least squares}
\acrodef{wsn}[WSN]{wireless sensor network}
\acrodef{wrt}[w.r.t.]{with respect to}

\newcommand{\Mmat}{\mathbf{M}}

\usepackage{catchfilebetweentags,etoolbox}
\makeatletter
\patchcmd{\CatchFBT@Fin@l}{\endlinechar\m@ne}{}%
  {}{\typeout{CFBT patch failed!}}%
\makeatother

\graphicspath{ {./Images/} }
\newtheorem{theorem}{Theorem}
\newtheorem{lemma}{Lemma}
\newtheorem{proposition}{Proposition}
\newtheorem{claim}{Claim}
 \newtheorem{definition}{Definition}

\usepackage[user,abspage]{zref} 
\makeatletter
\zref@newprop{col}{}              
\zref@addprop{main}{abspage}      
\zref@addprop{main}{col}          
\newcommand{\marksnippet}[1]{%
  \begingroup
    \edef\snip@col{\if@twocolumn \if@firstcolumn left\else right\fi \else single\fi}%
    \zref@setcurrent{col}{\snip@col}%
    \zlabel{snip:#1}%
  \endgroup
}
\makeatother
\title {Efficient Sampling Allocation Strategies for General Graph-Filter-Based Signal Recovery}
\author{Lital Dabush \IEEEmembership{Student Member, IEEE} and
Tirza Routtenberg, \IEEEmembership{Senior Member, IEEE}
\thanks{{\footnotesize{Lital Dabush and Tirza Routtenberg are with the School of Electrical and Computer Engineering, Ben-Gurion University of the Negev, Beer-Sheva 84105, Israel, e-mail: litaldab@post.bgu.ac.il, tirzar@bgu.ac.il.}}}
\thanks{This research was supported by the ISRAEL SCIENCE FOUNDATION (Grant No. 1148/22) and by the Israel Ministry of National Infrastructure and Energy.  L. Dabush is a fellow of the AdR Women Doctoral Program.
 }
 \vspace{-0.5cm}
}

\begin{document}	
	 \maketitle

\begin{abstract}
Sensor placement plays a crucial role in graph signal recovery in underdetermined systems. 
In this paper, we present the graph-filtered regularized maximum likelihood (GFR-ML) estimator of graph signals, which integrates general graph filtering with regularization to enhance signal recovery performance under a limited number of sensors.  
Then, we investigate task-based sampling allocation aimed at minimizing the mean squared error (MSE) of the GFR-ML estimator by wisely choosing sensor placement.  Since this MSE depends on the unknown graph signals to be estimated, we propose four cost functions for the optimization of the sampling allocation: the biased Cram$\acute{\text{e}}$r-Rao bound (bCRB), the worst-case MSE (WC-MSE), the Bayesian MSE (BMSE), and the worst-case BMSE (WC-BMSE), where the last two assume a Gaussian 
prior. We investigate the properties of these cost functions and develop two algorithms for their practical implementation: 1) the straightforward greedy algorithm;
and 2) the alternating projection gradient descent (PGD) algorithm that reduces the computational complexity. 
Simulation results on synthetic and real-world datasets of the IEEE 118-bus power system and the Minnesota road network demonstrate that, in the tested scenarios,  the proposed sampling allocation methods reduce the MSE by up to $50\%$ compared to the common sampling methods A-design, E-design, and LR-design. 
Thus, the proposed methods improve the estimation performance and reduce the required number of measurements in graph signal processing (GSP)-based signal recovery in the case of underdetermined systems.
\end{abstract}


\vspace{-0.25cm}
\section{Introduction}
Modern systems, \textcolor{black}{such as power or transportation networks,} that are naturally modeled as graph signals \cite{Newman_2010}, often generate \textcolor{black}{high-dimensional} signals supported on irregular structures. 
Recovering these signals from partial, corrupted, or noisy observations \textcolor{black}{is a fundamental task in \ac{gsp}}, with applications in  sensor network alignment \cite{dabush2023state,danilo_sampling},  time synchronization in distributed networks \cite{singer2011three,4177758}, and power system monitoring
\cite{drayer2020detection,dabush2023state,dabush2023verifying}. 
The performance of graph signal recovery depends heavily on which nodes are sampled  \cite{zhao2014identification}. 
{\textcolor{black}{As datasets grow in size and complexity,}} efficient sampling strategies {\textcolor{black}{that select informative subsets of nodes become}} crucial for reducing sensing, computation, and storage overhead \cite{Shuman_Ortega_2013,Sandryhaila2014}. 
{\textcolor{black}{In this context,  efficient \ac{gsp}-based task-driven sampling strategies under general filtering and regularization models are essential for scalable and accurate graph signal recovery.}}

\vspace{-0.25cm}
\subsection{\textcolor{black}{Related Works}}\label{sub_related_works}
The sampling and recovery of graph signals has gained growing interest in recent years \cite{anis2016efficient,Wang_2015,chen2015discrete,Wang_2016,Marques_2016}. 
Sampling methods in \ac{gsp} can be generally divided into random and deterministic approaches \cite{Tanaka2020Gene}.
In random sampling, nodes (vertices) are selected randomly according to a 
probability distribution, often designed to prioritize ``important”/``central" nodes \cite{Dinesh2022,perraudin2018global,pc_sampling,Random_sampling_bandlimited_signals}. 
 These methods are computationally efficient and easily implemented in a distributed manner, but often require many samples in order to achieve reconstruction quality comparable to that of deterministic sampling methods, even for bandlimited signals \cite{Tanaka2020Gene}.
{\textcolor{black}{Deterministic sampling methods, on the other hand, apply task-dependent criteria to select optimal sampling locations,  and typically yield better performance.}} As signal recovery is usually an ill-posed inverse problem, regularized or constrained optimization are commonly used to incorporate additional information and improve performance 
\cite{Tanaka2020Gene,
dabush2023state,Shuman_Chebyshev_Frossard_2011,ortega2022introduction,Wang_2015,Chen_2015_distributed,Wang_2016,Marques_2016,Segarra_2016}.

Most existing sampling methods for graph signals rely on specific structural assumptions, such as bandlimitedness or smoothness, to enable recovery. For the common assumed case of bandlimitedness of the graph signals, perfect recovery {\textcolor{black}{can be guaranteed under suitable}} sampling design regimes \cite{anis2016efficient,chen2015discrete,danilo_sampling}. Various sampling policies have been  developed  for this case, including A-optimality  (A-design)  \cite{anis2016efficient}, {\textcolor{black}{ E-optimality (E-design)}}~\cite{chen2015discrete},
fast distributed algorithms \cite{Wang_2015,Chen_2015_distributed}, local weighted measurements \cite{Wang_2016}, local aggregation \cite{Marques_2016}, and percolation from seeding nodes \cite{Segarra_2016}.
However, these methods are {\textcolor{black}{based on a strictly bandlimitedness assumption,}} which restricts their applicability in practical applications.
Smoothness priors offer a more flexible alternative {\textcolor{black}{ by assuming that signals vary only gradually between adjacent nodes. This has led to recovery methods based on }}Laplacian regularization and associated sampling strategies\cite{Shuman_Ortega_2013,Tanaka2020Gene,dabush2023state,Shuman_Chebyshev_Frossard_2011,ortega2022introduction}.
For example, a Laplacian-regularized sampling allocation (LR-design) 
was proposed in \cite{Bai2020} to minimize the largest eigenvalue of the estimator's gain matrix. 
Nonetheless, smoothness still imposes restrictive conditions that may not capture signal variability in practice~ \cite{chen2015discrete,Isufi2022sampling}.
{\textcolor{black}{Both bandlimitedness and smoothness can be }}
considered as special cases of graph filtering models for the regularization. 

\marksnippet{related_works1}
 {\textcolor{black}{Recent works consider more advanced settings, including randomized local aggregations for improved robustness to noise without requiring signal support knowledge \cite{sampling_randomized_aggregations_2019}, successive aggregation strategies based on the graph fractional Fourier transform \cite{WEI2023103970}, and multi-dimensional generalized sampling schemes that exploit subspace and smoothness priors \cite{WEI2024109601}. Sampling models in generalized shift-invariant spaces have also been developed for analog sparse signal domains \cite{Compressed_Sampling_FrFT_2021}. The generalized sampling framework in \cite{Tanaka_Eldar_2020} unifies subspace and smoothness priors without requiring strict bandlimitedness. }}
 Dictionary learning of sparse graph signals was proposed in~\cite{Isufi2022sampling}, but this method requires large training datasets. Other sampling techniques assume general measurement models and \ac{crb}-based sampling
~\cite{4663892,6981988,9967316}. Nevertheless, these methods do not incorporate \ac{gsp}  tools, such as graph filters or common Laplacian-based regularization. {\textcolor{black}{The recent developments motivate the need for flexible sampling methods that can handle various generalized filtering models and graph regularization within a unified framework for sampling design.}}

\vspace{-0.25cm}
\subsection{\textcolor{black}{Contributions}}\label{sub_contributions}
\vspace{-0.05cm}
In this paper, we develop a general framework for task-based sampling for the estimation of graph signals   {\textcolor{black}{that incorporates generalized graph filters in both the observation and the regularization models.}} First, we introduce the \ac{gfr}-\ac{ml} estimator, which integrates general graph filtering as a regularization for graph signal recovery. 
In order to optimize the sampling scheme, the  \ac{mse} of the \ac{gfr}-\ac{ml} estimator is a natural choice for the objective function. However, we show that this \ac{mse} is a function of the unknown parameters,  making direct optimization impractical. Thus, we develop sampling allocation strategies based on four cost functions designed to optimize estimation performance: 
(i) the non-Bayesian \ac{bcrb} (with the bias of the \ac{gfr}-\ac{ml} estimator); (ii) the \ac{wcmse}; (iii) the \ac{bmse}; (iv) and the \ac{wcbmse}.
To minimize these cost functions, we propose two practical algorithms: (i) a greedy {\textcolor{black}{heuristic}} algorithm  {\textcolor{black}{based on marginal gain}};
and (ii) an alternating \ac{pgd} algorithm based on the convex relaxation of the sampling allocation problem.
For the alternating \ac{pgd} algorithm, we derive the gradient expressions of all cost functions for general graph filters {\textcolor{black}{enabling computationally efficient and scalable implementation for sampling set selection}}. 
Simulation results on synthetic and real-world datasets, including the IEEE 118-bus power system \cite{iEEEdata} and the Minnesota road network \cite{perraudin2014gspbox}, demonstrate that the proposed sampling allocation methods outperform the A-design \cite{anis2016efficient}, E-design \cite{chen2015discrete}, and LR-design \cite{Bai2020} sampling strategies in sense of \ac{mse} of the resulting graph signal recovery. 

\vspace{-0.2cm}
\subsection{\textcolor{black}{Organization and Notations}}
\label{org_and_not_subsection}
\vspace{-0.05cm}
The remainder of the paper is organized as follows. In Section \ref{GSP_sec}, we provide background on \ac{gsp}. Section \ref{sec;model_est_prob} describes the model, the estimation approach, and the considered sampling allocation problem. In Section \ref{Sec;proposed_approach}, we propose and analyze different cost functions for the sampling approach: the \ac{bcrb}, \ac{wcmse}, \ac{bmse}, and \ac{wcbmse} 
approaches. In Section \ref{solvers_subsection}, we present a greedy algorithm and an alternating \ac{pgd} algorithm for efficient sensor location selection. 
Section \ref{sec;sim} presents our simulation study. Finally, the conclusions are provided in Section \ref{Conclusions}.

In the rest of this paper, vectors and matrices are denoted by boldface lowercase and uppercase letters, respectively. 
The notations $(\cdot)^T$, $(\cdot)^{-1}$, $(\cdot)^{\dagger}$, 
and ${\text{tr}}(\cdot)$ denote the transpose, inverse, Moore-Penrose pseudo-inverse, and trace operators,
respectively. 
The $m$th element of the vector $\avec$  and the $(m, q)$th element of the matrix $\Amat$ are denoted by $a_m$ and $A_{m,q}$, respectively. The parameters $\lambda_i$, $\lambda_{\text{min}}(\Amat)$ and $\lambda_{\text{max}}(\Amat)$ denote the $i$th, the minimum and the maximum eigenvalues of the matrix $\Amat$, respectively. 
  The gradient of a scalar function $g(\avec) \in \mathbb{R}$ \ac{wrt} the vector $\avec \in \mathbb{R}^{M \times 1}$ is denoted by $\nabla_{\avec} g(\avec) \in \mathbb{R}^{M \times 1}$. The Jacobian of a vector function $\gvec(\avec) \in \mathbb{R}^{K \times 1}$ \ac{wrt} $\avec$ is denoted by $\nabla_{\avec} \gvec(\avec) \in \mathbb{R}^{K \times M}$, where each entry is defined as $\big[\nabla_{\avec} \gvec(\avec)\big]_{m,k}=\frac{\partial \mathrm{g}_{m}(\avec)}{\partial a_{k}}$.
  $\Imat$, and
 $\onevec$ and $\zerovec$ denote the identity matrix, and
 vectors of ones and zeros, respectively, with appropriate dimensions, and  
 $||\cdot||$ denotes the Euclidean $l_2$-norm of a vector. 
For a vector $\avec$, $\mathcal{P}_{\mathcal{C}}(\avec)$ denotes the  projection of $\avec$ onto the set $\mathcal{C}$, and ${\text{diag}}(\avec)$ is a diagonal matrix whose $(m,m)$th entry is $a_m$.

\section{Background: Graph Signal Processing (GSP)}
\label{GSP_sec}
Let ${\mathcal{G}}({\mathcal{V}},\xi)$ be a general undirected  weighted graph,
where ${\mathcal{V}}=\{1,\ldots,N\}$ and ${\xi}$ are the sets of nodes and edges, respectively.
The matrix $\Wmat\in{\mathbb{R}}^{N\times N}$ is the weighted adjacency matrix of the graph $\mathcal{G}({\mathcal{V}},\xi)$, where $W_{k,n}\geq 0$ denotes the weight of the edge between node $k$ and node $n$, and 
$W_{k,n}= 0$ if no  edge exists between $k$ and $n$.
\marksnippet{L_def1}
The \textcolor{black}{$(k,l)$th element of the} graph Laplacian matrix \textcolor{black}{is defined as 
\begin{equation}
\label{Def_lap_eq}
L_{k,l} = \begin{cases}
	\displaystyle \sum_{n=1}^N  W_{k,n},& k = l \\
  \displaystyle  -W_{k,l},& {\text{otherwise}} 
\end{cases},~~~k,l=1,\ldots,N.
\end{equation}}
\textcolor{black}{The Laplacian matrix}
is a real positive semi-definite matrix with the  \ac{evd} defined as 
\be
\Lmat=\Vmat{{\text{diag}}}(\pmb{\lambda}) \Vmat^{-1},   \label{SVD_new_eq}
\ee
where the columns of $\Vmat$ are the eigenvectors of $\Lmat$, $\Vmat^{T}=\Vmat^{-1}$, and $\pmb{\lambda} \in {\mathbb{R}}^{N}$ is a vector of the ordered eigenvalues of $\Lmat$ in  decreasing order.
We assume that ${\mathcal{G}}({\mathcal{V}},\xi)$ is a connected graph, and thus, $\lambda_2\neq 0$ \cite{Newman_2010}.
By analogy to the  frequency of signals in \ac{dsp},
 the Laplacian eigenvalues,
$\lambda_1,\ldots,\lambda_N$, can be interpreted as the graph frequencies.
Together with the eigenvectors in
$\Vmat$, they define the spectrum of the graph 
\cite{Shuman_Ortega_2013}.

A graph signal is a function that resides on a graph, assigning a scalar value to each node.
The \ac{gft} of a graph signal $\avec\in{\mathbb{R}}^N$  \ac{wrt} the graph ${\mathcal{G}}({\mathcal{V}},\xi)$ is defined as  \cite{Shuman_Ortega_2013,Sandryhaila2014}
 \begin{equation}
\label{GFT}
\tilde{\avec}\triangleq \Vmat^{-1}\avec.  
 \end{equation}
  Similarly,  the inverse \ac{gft} is obtained by a left multiplication of $\tilde{\avec}$ by $\Vmat$.
The \ac{tv} of a graph signal $\avec$ 
satisfies
\begin{equation}
\label{eq:Dirichlet energy}
{\avec^T}\Lmat\avec 
=%
\frac{1}{2} \mathop \sum _{k=1}^N\sum_{n=1}^N W_{k,n}\big(a_k - a_n\big)^2
=\sum_{n=1}^N\lambda_n\tilde{a}_n^2,
\end{equation} where  the first equality is obtained by substituting  \eqref{Def_lap_eq}, and the second equality is obtained 
by substituting \eqref{SVD_new_eq} and \eqref{GFT}.
 
The \ac{tv} from \eqref{eq:Dirichlet energy}  
is a smoothness measure, which is used in graphs to quantify changes \ac{wrt} the variability that is encoded by the weights of the graph \cite{Shuman_Ortega_2013,Chen_Kovacevic_2015}.
A graph signal, $\avec$, is smooth if $\avec^T\Lmat\avec\leq\varepsilon$, 
where $\varepsilon$ is small in terms of the specific application  \cite{Shuman_Ortega_2013}.
Thus, the smoothness assumption implies that neighboring nodes have similar values, 
and the graph signal spectrum is in the small eigenvalues region (see \eqref{eq:Dirichlet energy}).

Linear and shift-invariant graph filters play essential roles in \ac{gsp}.
These filters generalize linear time-invariant filters 
used 
in \ac{dsp}, and enable processing over graphs \cite{8347162,Shuman_Ortega_2013}.
A Laplacian-based graph filter can be defined in the graph frequency domain as  a function $h(\cdot)$ that  allows an \ac{evd}  \cite{8347162}:
\marksnippet{eig_filter1}
\begin{equation}\label{graph_filter} h(\Lmat)\hspace{-0.05cm}\define\hspace{-0.05cm}\Vmat{{\text{diag}}}(h(\pmb{\lambda}))\Vmat^{-1},~h(\pmb{\lambda})=[h(\lambda_1),\dots,h(\lambda_{N})]^{\textcolor{black}{T}}, \hspace{-0.05cm}\end{equation}
  where   $h(\lambda_n)$ is the graph  filter  frequency response  at the graph frequency $\lambda_n$, $n=1,\ldots,N$.
 The graph filter frequency response should be identical for all equal eigenvalues (see, e.g. \cite{ortega2022introduction}, Chapter 3).
A graph filter applied on a graph signal is a linear operator   that satisfies the following:
\begin{equation}
\label{a_out_a_in}
   \avec^{(\rm{out})}=h(\Lmat)\avec^{(\rm{in})},
\end{equation}
where $ \avec^{(\rm{out})}$ and $\avec^{(\rm{in})}$ are the output and input graph signals. 
Following \ac{dsp} conventions, \acp{lpf} are filters that do not significantly affect the frequency content of low-frequency signals but attenuate the magnitude of high-frequency signals. Analogously, \acp{hpf} pass high-frequency signals while attenuating low frequencies \cite{Sandryhaila2014}.

\section{Model, estimators, and problem formulation}\label{sec;model_est_prob}
In this section, we describe the sampling allocation problem in \ac{gsp}-based models with the goal of enhancing estimation performance.
First, 
we introduce the measurement model in Subsection \ref{subsec;model}. Then,
we derive the \ac{gfr}-\ac{ml} estimator associated with this model in Subsection \ref{subsec;estimator}.
Finally, we formulate the sampling allocation problem associated with the estimation performance in Subsection \ref{subsec;problemF}. 

\vspace{-0.2cm}
\subsection{GSP Measurement Model}\label{subsec;model}
\vspace{-0.05cm}
We consider a linear graph filtering  model:
\begin{equation}
\yvec=h_{\text{M}}(\Lmat){\xvec}+\evec,
\label{original_model}
\end{equation}
where $\yvec\in\mathbb{R}^N$ and $\xvec\in\mathbb{R}^N$ are the output and input graph signals, respectively. The vector $\evec\in {\mathbb{R}}^N$  
represents a zero-mean Gaussian  noise signal 
with covariance matrix  $\Rmat$, 
 i.e. $\evec \sim \mathcal{N}(\zerovec, \Rmat)$.
 \marksnippet{define_h_m1}
The graph filter, $h_{\text{M}}(\Lmat)\in {\mathbb{R}}^{N\times N}$, \textcolor{black}{where the subscript $M$ indicates that this graph filter models the measurement process,} is assumed to be known along with the graph topology represented by $\Lmat$. 
The input graph signal, $\xvec$, is unknown and needs to be estimated.
The model in \eqref{original_model} is well-established in \ac{gsp} and has been utilized  in various applications \cite{water_gsp,1583238, Sandryhaila2014,dong2020graph,Tanaka2020Gene,Vassilis_2016}.  
This model effectively captures how signals propagate over network structures, leveraging graph filters to represent signal behavior dictated by underlying graph topologies \cite{Shuman_Ortega_2013}.

In practice, resource constraints such as budget, energy, and maintenance often limit the number of sensors deployed in large networks \cite{water_gsp,Tanaka2020Gene,zhao2014identification};  this is since sensor deployment incurs significant costs, including initial investment and ongoing expenses for data transmission and power supply. 
To model this situation, let $\mathcal{S} \subseteq \mathcal{V}$ denote the subset of nodes selected for sensor deployment (sampling set), where $\mathcal{V}$ is the set of all nodes with $|\mathcal{V}| = N$. 
The measurement model from \eqref{original_model} under this setting of partial observations is 
\be
\label{reduced_model_new}
    \yvec_{\mathcal{S}} =  \left[h_{\text{M}}(\Lmat)\right]_{\mathcal{S},\mathcal{V}} \xvec +  \evec_{\mathcal{S}},
\ee
where $\yvec_{\mathcal{S}} \in \mathbb{R}^{|\mathcal{S}|}$ represents the observed measurements at the sampled nodes, and $\evec_{\mathcal{S}} \in \mathbb{R}^{|\mathcal{S}|}$ is the corresponding noise vector. The notation $[h_{\text{M}}(\Lmat)]_{\mathcal{S},\mathcal{V}}$ denotes the sub-matrix of $h_{\text{M}}(\Lmat)$  containing the rows indexed by $\mathcal{S}$ that  correspond to the sampled nodes.
An alternative formulation uses a sampling indicator vector $\dvec\in \{0,1\}^N$, where  $d_n = 1$ if the $n$th node is sampled and $d_n = 0$ otherwise. Then, the partial measurement model from \eqref{reduced_model_new}
can be written using $\Dmat = {\text{diag}}(\dvec)$ as
\begin{equation}
    \Dmat \yvec = \Dmat h_{\text{M}}(\Lmat) \xvec + \Dmat \evec.
    \label{reduced_model_diag}
\end{equation}

For the system in \eqref{reduced_model_new} (or, equivalently,  \eqref{reduced_model_diag}) to 
yield a unique solution for $\xvec$ in the least squares sense even in the noiseless scenario, 
\marksnippet{column_rank_Cond1}
the matrix $\left[h_{\text{M}}(\Lmat)\right]_{\mathcal{S},\mathcal{V}}$ ($\Dmat h_{\text{M}}(\Lmat)$) must be full \textcolor{black}{column} rank. 
 Otherwise, the system is underdetermined, and additional information is needed to obtain a unique solution for $\xvec$. 
As mentioned above, a common approach in \ac{gsp} to address this issue is to assume that the graph signals are smooth or have low \ac{tv}. 
To extend  this assumption to other types of graph signals beyond low-pass or smooth signals,
\marksnippet{definition_hr1}
we introduce a constraint \ac{wrt} a general  semi-definite graph filter 
$h_{\text{R}}^{+}(\Lmat)\in \mathbb{R}^{N \times N}$,  \textcolor{black}{ where the subscript $R$ indicates its role in regularization, and the superscript  `+' denotes that the filter is required to be a positive semi-definite matrix.} 
That is,  
we assume  
 \begin{equation}
\label{smothness_full_theta}
{\mathcal{E}}_{\Lmat}(\xvec)=
{(\xvec-\xvec_0)^T}h_{\text{R}}^{+}(\Lmat)(\xvec-\xvec_0)\leq\varepsilon,
 \end{equation}
where $\varepsilon$ is a tolerance parameter, and
$\xvec_0\in \mathbb{R}^N$ is a reference signal. 

Unlike prior formulations that assumed $\xvec_0 = \zerovec$ (see, e.g.~\cite{Random_sampling_bandlimited_signals} and p.~132 of~\cite{ortega2022introduction}), our framework allows a general reference point $\xvec_0$. This enables the regularization to capture prior information centered at an arbitrary location, rather than implicitly constraining $\xvec$ to lie near the kernel of $h_{\text{R}}^{+}(\Lmat)$. This flexibility is especially important when $h_{\text{R}}^{+}(\Lmat)$ is full rank, since in that case the kernel includes only the zero vector. Additionally, while previous works focused on \acp{hpf} as regularizers, our approach accommodates arbitrary graph filters for $h_{\text{R}}^{+}(\Lmat)$ for regularization, broadening the applicability of the regularization framework.

\marksnippet{r2_notations_generalize_graph_filter_regularization1}
The assumption embodied in~\eqref{smothness_full_theta} encompasses common special cases: 
\color{black}
\subsubsection{Special Case 1 - Smooth Graph Signal Estimation}\label{sc_smooth}
\color{black}
By  defining 
$h_{\text{R}}^{+}(\Lmat)=\Lmat$ \textcolor{black}{ and setting $\xvec_0$ to lie in the kernel of $h_{\text{R}}^{+}(\Lmat)$, the regularization in \eqref{smothness_full_theta} reduces to the classical smoothness condition $\xvec^T\Lmat\xvec\leq\varepsilon$.} This smoothness assumption is widely used in various applications (see, e.g.,  \cite{ortega2022introduction,Tanaka2020Gene,Shuman_Chebyshev_Frossard_2011,dabush2023state}), and is especially beneficial for denoising and reconstruction tasks \cite{ortega2022introduction,Shuman_Ortega_2013,Random_sampling_bandlimited_signals,Shuman_Chebyshev_Frossard_2011}. 
In semi-supervised learning, for example, $\xvec$ represents label scores, $\Lmat$ encodes data similarity, and the smooth regularization enables smooth propagation of labels across the graph, thereby improving classification performance~\cite{Shuman_Ortega_2013,8347162}.

 \marksnippet{special-cases-bandlimited-setting1}
\subsubsection{\textcolor{black}{Special Case 2 - Strictly Graph-Bandlimited Signal Estimation}}\label{special_case_bandlited}
 Bandlimited graph signals confined to a subset of graph frequencies $\mathcal{R}\subset\mathcal{V}$, can be modeled 
\color{black}
 by specifying the regularizer $h_{\text{R}}^{+}(\Lmat)$, prior mean $\xvec_0$, and tolerance parameter $\varepsilon$, as follows:
\begin{equation}
\label{special_case}
h_{\text{R}}^{+}(\lambda_i)=0, \forall i\in\mathcal{R},~~h_{\text{R}}^{+}(\Lmat)\xvec_0=\zerovec,~~\varepsilon=0.
\end{equation}
Under these settings, 
the constraint in \eqref{smothness_full_theta} eliminates the estimator's graph frequencies outside $\mathcal{R}$, i.e.  $\hat{\tilde{\xvec}}_{\mathcal{V}\setminus\mathcal{R}}=\zerovec$. 
\color{black}
{\textcolor{black}{In sections \ref{ss;bandlimited_estimation} and \ref{appendix_relation_bandlimited} of the supplementary material attached to this paper, it is shown that the parameter settings in \eqref{special_case}  leads to the widely-use approaches for sampling and recovery of bandlimited graph signals in \ac{gsp}; In particular, for $h_{\text{M}(\Lmat)}=\Imat$ the estimation and sampling coincide with those in }} \cite{Sandryhaila2014,smoothVSbandlimited,reg_formulation}. 


These priors capture domain knowledge in diverse applications. For example, 
in image processing, $\xvec$ may denote pixel intensities, while and $\Lmat$ is defined by the pixel adjacency. 
In power systems, where $\xvec$ represents the states (voltages) and $\Lmat$ corresponds to the admittance matrix, the prior  in
\eqref{smothness_full_theta}  imposes smoothness  
with $\xvec_0$ as a reference voltage and $h_{\text{R}}^{+}(\Lmat)$ is any \ac{hpf}  \cite{drayer2020detection,dabush2023state}. 
 Similarly, in sensor networks,  small variations \ac{wrt} the graph can be exploited for fault detection by applying a \ac{hpf} to the 
 graph signals \cite{Sandryhaila2014}.

 Another interpretation of the prior in \eqref{smothness_full_theta} is from a Bayesian estimation perspective, where the input graph signal, $\xvec$,  is random and has a 
 Gaussian distribution \cite{Dong_Vandergheynst_2016,ramezani2019graph}:
\begin{equation} \label{x_distribution}
		\xvec \sim \mathcal{N}(\xvec_0,
  \frac{1}{\mu}(h_{\text{R}}^{+}(\Lmat))^{\dagger}),
\end{equation}
where $\mu$ scales the variance of the signal,  
\color{black}
and $\xvec_0$ is not composed from vectors in the  kernel of $h_{\text{R}}^+(\Lmat)$. 
\color{black}
The log-likelihood of this prior (up to a constant) is $\mu{(\xvec-\xvec_0)^T}h_{\text{R}}^{+}(\Lmat)(\xvec-\xvec_0)$, which can be seen as a regularization.

\subsection{GFR-ML Estimator}
\label{subsec;estimator}
To incorporate the prior in \eqref{smothness_full_theta} into the estimation process, we formulate an optimization problem that combines the likelihood function of the measurement vector from \eqref{reduced_model_diag} with the constraint from \eqref{smothness_full_theta}.
The resulting estimator is obtained by solving the following optimization problem:
\begin{equation}
    \hat{\xvec} = \arg \min_{\xvec\in{\mathbb{R}}^N}  \left\| \Dmat (\yvec - h_{\text{M}}(\Lmat) \xvec) \right\|_{\Rmat^{-1}}^2
    + \mu\left\|\xvec- \xvec_0\right\|_{h_{\text{R}}^{+}(\Lmat)}^2,
    \label{eq:optimization_problem}
\end{equation}
where the operator $\left\| \zvec \right\|_{\Amat}^2 \define \zvec^T \Amat \zvec$ denotes the quadratic form of the matrix $\Amat$, and $\mu > 0$ is a regularization parameter that balances the data fidelity term with the prior. \textcolor{black}{ The regularization parameter $\mu$ can be explicitly related to the constraint threshold $\varepsilon$ in~\eqref{smothness_full_theta} via the \ac{kkt} conditions. Specifically, $\mu$ and $\varepsilon$ are inversely related: a larger $\mu$ corresponds to a smaller $\varepsilon$, with $\varepsilon = 0$ implying $\mu \to \infty$. 
}

Solving the optimization problem in \eqref{eq:optimization_problem} leads to a regularized \ac{ml} estimator, named here the \ac{gfr}-\ac{ml} estimator, for the recovery of $\xvec$. Since the objective in \eqref{eq:optimization_problem} is a convex function of $\xvec$, by equating the derivative of \eqref{eq:optimization_problem} \ac{wrt} $\xvec$ to zero to find the extremum (see, e.g.  p. 17 in \cite{vanwieringen2020lecture}), we obtain
\beqna
\label{estimator}
\hat{\xvec}=\Kmat^{-1}(\dvec)({h}_{\text{M}}(\Lmat)\Dmat\Rmat^{-1}\Dmat\yvec+\mu{h}_{\text{R}}^{+}(\Lmat)\xvec_0),
\eeqna
where 
\be
\label{estimator_matrix}
\Kmat(\dvec)
\define {h}_{\text{M}}(\Lmat)\Dmat\Rmat^{-1}\Dmat{h}_{\text{M}}(\Lmat)+\mu{h}_{\text{R}}^{+}(\Lmat).
\ee 
It should also be noted that for the Bayesian perspective described in \eqref{x_distribution}, the estimator in \eqref{estimator}-\eqref{estimator_matrix} is also the \ac{mmse} estimator.

The inclusion of $\mu h_{\text{R}}^{+}(\Lmat)$ 
allows $\Kmat(\dvec)$ to be positive definite and invertible, even in  underdetermined systems, where $|\mathcal{S}| < N$
 and  ${h}_{\text{M}}(\Lmat)\Dmat\Rmat^{-1}\Dmat{h}_{\text{M}}(\Lmat)$ is a rank-deficient matrix. This regularization leads to the possibility of a unique solution and incorporates prior knowledge into the estimation process. The \ac{gfr}-\ac{ml} estimator uses graph filters, given by ${h}_{\text{M}}$ and ${h}_{\text{R}}^{+}$, in both the measurement model and the regularization term. Unlike estimators that rely solely on Laplacian-based regularization, it employs the flexibility of general graph filters via $h_{\text{R}}^{+}(\Lmat)$, enabling adaptation to specific structural and spectral properties, such as bandlimitedness and task-specific frequency weighting.

Furthermore, as both the measurement model $h_{\text{M}}(\Lmat)$ and the regularization term $h_{\text{R}}^{+}(\Lmat)$ are graph filters, these filters can be efficiently approximated or implemented as finite impulse response (FIR) graph filters. Specifically, when $h_{\text{M}}(\Lmat)$ and $h_{\text{R}}^{+}(\Lmat)$ are modeled as polynomial graph filters \cite{ortega2022introduction}, i.e. $h_{\text{M}}(\Lmat) = \sum_{k=0}^{K_{\text{M}}} a_k^{(\text{M})} \Lmat^k$ and $h_{\text{R}}^{+}(\Lmat) = \sum_{k=0}^{K_{\text{R}}} a_k^{(\text{R})} \Lmat^k$, \textcolor{black}{where $\{a_k^{(\text{M})}\}_{k=0}^{K_{\text{M}}}$ and  $\{a_k^{(\text{R})}\}_{k=0}^{K_{\text{R}}}$ are the coefficients of $h_{\text{M}}(\Lmat)$ and $h_{\text{R}}^{+}(\Lmat)$,  respectively,} the matrix-vector multiplications required for the estimator involve only localized computations, where each node processes information from its immediate neighbors. In order to increase the efficiency, Chebyshev polynomial approximations \cite{Shuman_Chebyshev_Frossard_2011} can be employed to reduce computational complexity.  These properties facilitate scalable, distributed implementations, making the \ac{gfr}-\ac{ml} estimator particularly suitable for large-scale networks.


\subsection{Problem Formulation - Sampling Allocation}\label{subsec;problemF}
The sensor locations selected in the sampling step 
have a significant impact on the estimation performance in various applications (see, e.g. \cite{zhao2014identification,wsn_data}).  
We assume a constrained amount of sensing resources, $\sum_n d_n=q$, e.g. due to limited energy and communication budget. This requirement can be rewritten as the constraint $\|\dvec\|^2=q$. Thus, the sampling task can be written as
\beqna\label{sampling_statment}
\dvec^{opt}=\arg\min_{\dvec\in\{0,1\}^N:\|\dvec\|^2=q}C(\dvec),
\eeqna
where $C(\dvec)$ denotes a general cost function that is associated with the estimation performance.
Alternatively, if the number of sensors to be deployed, i.e. $q$, is unknown,  we can reformulate the sensor location selection problem to minimize the number of nonzero entries in 
$\dvec$ rather than fixing $q$ as in \eqref{sampling_statment}. This approach yields the number of deployed  sensors as a byproduct, as follows:
\begin{align}\label{eq_l0}
    \dvec^{opt} &= \arg\min_{\dvec\in\{0,1\}^N} \|\dvec\|^2  \\
    \text{s.t.} & \quad C(\dvec)\leq \varepsilon.\nonumber 
\end{align}

As our goal is to minimize the \ac{mse} of the \ac{gfr}-\ac{ml} estimator, the natural choice for 
$C(\dvec)$  is the \ac{mse} of this non-Bayesian estimator, $\hat{\xvec}$, given by
\beqna
\label{MSE1}
{\text{MSE}}(\hat{\xvec})={\rm{E}}[(\hat{\xvec}-\xvec)^T(\hat{\xvec}-\xvec);\xvec],
\eeqna
where 
${\rm{E}}[\cdot;\xvec]$ denotes the expectation  parametrized by  the deterministic parameter vector, $\xvec$. 
The estimation error vector of the \ac{gfr}-\ac{ml} estimator, $\hat{\xvec}-\xvec$, is obtained by  substituting \eqref{reduced_model_diag}, \eqref{estimator},  and \eqref{estimator_matrix}, and using the identity \[\Imat-(\Amat+\Bmat)^{-1}\Amat=(\Amat+\Bmat)^{-1}\Bmat \] with $\Amat={h}_{\text{M}}(\Lmat)\Dmat\Rmat^{-1}\Dmat{h}_{\text{M}}(\Lmat)$  and $\Bmat=\mu{h}_{\text{R}}^{+}(\Lmat)$, which results in 
\beqna
\label{error_vec}
    \hat{\xvec}-\xvec=\Kmat^{-1}(\dvec)
({h}_{\text{M}}(\Lmat)\Dmat\Rmat^{-1}\Dmat\evec+{h}_{\text{R}}^{+}(\Lmat)(\xvec_0-\xvec)).
\eeqna
By substituting \eqref{error_vec} in \eqref{MSE1} and using the fact that 
the covariance of $\Dmat\evec$ is $\Dmat \Rmat \Dmat$, we obtain that the 
 \ac{mse} of the \ac{gfr}-\ac{ml} estimator is 
\beqna\label{mse_general}
\mathrm{E}[(\hat{\xvec}-\xvec)^T(\hat{\xvec}-\xvec);\xvec]=\mu^2\left\|\Kmat^{-1}(\dvec)h_{\text{R}}^{+}(\Lmat)(\xvec- \xvec_0)\right\|^2\hspace{0.05cm}\nonumber\\+{\text{tr}}\Big(\Kmat^{-1}(\dvec)
{h}_{\text{M}}(\Lmat)\Dmat\Rmat^{-1}\Dmat{h}_{\text{M}}(\Lmat)\Kmat^{-1}(\dvec)\Big).
\eeqna
\marksnippet{r1c4-dependecy-on-x1}
{
{\textcolor{black}{In the general case,}} the \ac{mse} of the \ac{gfr}-\ac{ml} estimator in \eqref{mse_general} is a function of the unknown input graph signal, $\xvec$.
This dependency arises from the absence of an assumed prior distribution for \( \xvec \), which precludes averaging the \ac{mse} over a known distribution. As a result, the bias term in the \ac{mse} expression (first term on the r.h.s. of \eqref{mse_general}) remains signal-dependent.
 Since sensor placement is typically determined at the deployment stage and cannot adapt to each realization of \( \xvec \), optimizing based on such a signal-dependent \ac{mse} is impractical {\textcolor{black}{in the general case}}, and thus, it cannot be used directly as the cost function \( C(\dvec) \). Hence,  alternative cost functions must be considered.

 \color{black}
 Nevertheless, there exist special cases where the dependency on \( \xvec \) vanishes.
 For example, in the case of strictly bandlimited graph signals discussed in Subsection~\ref{special_case_bandlited}, where \( \tilde{\xvec}_{V \setminus R} = \zerovec \) and \( h_{\text{R}}^{+}(\Lmat)(\xvec - \xvec_0) = \zerovec \), the \ac{mse} in \eqref{mse_general} becomes independent of \( \xvec \), as shown in the supplementary material.
  However, while enforcing such hard constraints enables tractable optimization, it limits robustness to model mismatch. Thus, even in these settings, cost functions other than the \ac{mse}  may offer improved generalization.
\color{black}
Therefore, {\color{black}to enable a general and robust sampling framework that accommodates arbitrary regularization (including biased estimators) and is resilient to model mismatch}, we propose alternative cost functions that approximate, or upper-bound, the \ac{mse} while avoiding explicit dependence on the unknown signal \( \xvec \).
}

\section{Proposed sampling allocation}\label{Sec;proposed_approach}
In this section, we introduce four possible cost functions $C(\dvec)$ for the sampling allocation scheme:
 the \ac{bcrb} (Subsection \ref{subsec;biased_crb}),  \ac{wcmse}  (Subsection \ref{subsec;worstMSE}), \ac{bmse}  (Subsection  \ref{subsec;BayesianMSE}),  and  \ac{wcbmse} (Subsection \ref{subsec;worstBayesianMSE}). General remarks on the relations between these cost functions are provided in Subsection \ref{subsec;remarks}.

\vspace{-0.25cm}
\subsection{Cost Function 1: bCRB}\label{subsec;biased_crb}
The first cost function is based on replacing the \ac{mse} from \eqref{mse_general} with the \ac{bcrb}~\cite{Kayestimation}. The \ac{bcrb} provides a lower bound on the \ac{mse} for any estimator with a given bias function\cite{Kayestimation}.
In our case, the distribution of the partial measurement vector obtained from the sensor subset $\mathcal{S}$, as described in \eqref{reduced_model_diag},  is as follows:
\beqna\label{pdf}
\Dmat\yvec\sim\mathcal{N}(\Dmat{h}_{\text{M}}(\Lmat)\xvec,\Rmat).
\eeqna
The \ac{bcrb} on the trace of the \ac{mse} for this Gaussian model
is given by
(see, e.g. pp. 45-46 in \cite{Kayestimation})
\beqna\label{bias_CRB}
 {\text{bCRB}}({\dvec})\define
 {\text{tr}}\Big((\Imat+\nabla_{\xvec} \bvec(\xvec,\dvec))\hspace{3cm}\nonumber\\
\times({h}_{\text{M}}(\Lmat)\Dmat\Rmat^{-1}\Dmat{h}_{\text{M}}(\Lmat))^\dagger(\Imat+\nabla_{\xvec} \bvec(\xvec,\dvec))^T\Big),
\eeqna
where
$\nabla_{\xvec} \bvec(\xvec,\dvec) \in \mathbb{R}^{N\times N}$ is the Jacobian matrix of the estimator bias,  defined as $\bvec(\xvec,\dvec)\define {\rm{E}}[\hat{\xvec}-\xvec]$.
Note that for $|\mathcal{S}|<N$, $\Dmat$ is not a full rank matrix. Consequently, the \ac{fim} is also not full rank, as it results from the multiplication of singular matrices. The use of the 
pseudo-inverse in \eqref{bias_CRB} enables the the option of 
 a singular \ac{fim} \cite{Hero1996Usman}. 
 
By using the model in \eqref{reduced_model_diag} and the estimator in \eqref{estimator}, we obtain that
the bias of the \ac{gfr}-\ac{ml} estimator 
 is
\begin{equation}\label{bias}
\bvec(\xvec,\dvec)=\Kmat^{-1}(\dvec)({h}_{\text{M}}(\Lmat)\Dmat\Rmat^{-1}\Dmat{h}_{\text{M}}(\Lmat){\xvec}+\mu{h}_{\text{R}}^{+}(\Lmat)\xvec_0)-{\xvec}.
\end{equation}
Thus, the gradient of the bias from \eqref{bias} \ac{wrt} $\xvec$ is
\beqna\label{bias_grad}
\nabla_{\xvec} \bvec(\xvec,\dvec)=\Kmat^{-1}(\dvec){h}_{\text{M}}(\Lmat)\Dmat\Rmat^{-1}\Dmat{h}_{\text{M}}(\Lmat)-\Imat.
\eeqna
By substituting  \eqref{bias_grad} in \eqref{bias_CRB}, and using the pseudo-inverse property $\Amat=\Amat\Amat^{\dagger}\Amat$, we obtain that the \ac{bcrb} on the \ac{mse}  of estimators with the \ac{gfr}-\ac{ml} bias  is given by
 \begin{equation}\label{CRB}
    {\text{bCRB}}({\dvec})= {\text{tr}}\Big(\Kmat^{-1}(\dvec)
{h}_{\text{M}}(\Lmat)\Dmat\Rmat^{-1}\Dmat{h}_{\text{M}}(\Lmat)\Kmat^{-1}(\dvec)\Big).
\end{equation}
It can be seen that for the special case of 
$\xvec = \xvec_0$, the \ac{mse} of the \ac{gfr}-\ac{ml} estimator from \eqref{mse_general} coincides with the \ac{bcrb} given in \eqref{CRB}. For the special case \textcolor{black}{of smooth graph signal estimation from Subsection \ref{sc_smooth} with } $h_{\text{M}}(\Lmat)=\Lmat$, 
the sampling scheme from our previous work \cite{dabush2023state} is obtained. 

\vspace{-0.2cm}
\subsection{Cost Function 2: WC-MSE}\label{subsec;worstMSE}
\vspace{-0.05cm}
As an alternative way to address the dependency of the \ac{mse} in \eqref{mse_general} on the unknown parameter $\xvec$, the following cost function is  proposed; it is based on a worst-case bias at  $\xvec$ that lies in the unit ball centered at $\xvec_0$, 
 $\mathcal{B}(\xvec_0)\define\{\forall \xvec\in\mathbb{R}^N~|~||\xvec-\xvec_0|| \leq 1\}$,  in a way that is 
similar to the rationale in \cite{eldar2004minimum}. 
We define the associated \ac{wcmse} of $\hat{\xvec}$ as
\beqna\label{worst_mse}
\text{MSE}_{WC}(\dvec) \define \max_{\xvec\in \mathcal{B}(\xvec_0)} \mathrm{E}[(\hat{\xvec}-\xvec)^T(\hat{\xvec}-\xvec)] \hspace{1.7cm}\nonumber\\
={\text{bCRB}}({\dvec})+\max_{\xvec\in \mathcal{B}(\xvec_0)} \mu^2\left\|\Kmat^{-1}(\dvec)h_{\text{R}}^{+}(\Lmat)(\xvec- \xvec_0)\right\|^2
\nonumber\\
 ={\text{bCRB}}({\dvec})
+ \mu^2 \lambda_{\text{max}}\Big({h}_{\text{R}}^{+}(\Lmat)\Kmat^{-2}(\dvec){h}_{\text{R}}^{+}(\Lmat)\Big),\hspace{0.9cm}
\nonumber\\
 ={\text{bCRB}}({\dvec})
+ \mu^2 \sigma^2_{\text{max}}\Big(\Kmat^{-1}(\dvec){h}_{\text{R}}^{+}(\Lmat)\Big),\hspace{1.9cm}
\eeqna
where the second equality is obtained by substituting \eqref{mse_general}  and \eqref{CRB},   the third equality is obtained by using the Rayleigh quotient theorem (\hspace{1sp}{\cite{Horn2012}}, pp. 234-235), and the last equality follows from the spectral norm identity 
$\lambda_{\text{max}}(\Amat^T\Amat)=\sigma_{\text{max}}^2(\Amat)$, 
with $\sigma_{\text{max}}(\cdot)$ denoting the largest singular value 
and $\Amat=\Kmat^{-1}(\dvec)h_{\text{R}}^{+}(\Lmat)$ (\hspace{1sp}{\cite{Horn2012}}, p. 346). 
The first term in \eqref{worst_mse} captures the contribution of the measurement noise, while the second term,  which is the spectral norm of $h_{\text{R}}^{+}(\Lmat)\Kmat^{-2}(\dvec, \mu)h_{\text{R}}^{+}(\Lmat)$ (\hspace{1sp}{\cite{Horn2012}}, p. 346), quantifies the worst-case impact of the regularization. 
The resulting \ac{wcmse} provides a robust metric for estimation under unfavorable conditions. The constraint $\xvec \in \mathcal{B}(\xvec_0)$ bounds 
deviations from the reference signal $\xvec_0$, while focusing on worst-case directional impacts. 
When $\xvec_0$ is reliable, this constraint models uncertainty within a meaningful range.

\vspace{-0.2cm}
\subsection{Cost Function 3: BMSE}\label{subsec;BayesianMSE}
\vspace{-0.05cm}
Until this subsection, the variable $\xvec$ has been treated as deterministic. In this subsection, we adopt a Bayesian perspective, as described in \eqref{x_distribution}, where 
\marksnippet{prior_distrbution1}
\(\xvec \sim \mathcal{N}(\xvec_0,\frac{1}{\mu}(\textcolor{black}{{h}_{\text{R}}^{+}(\Lmat)})^{\dagger})\)
. 
Then, we derive the theoretical minimum \ac{mse} for the Bayesian setting. 
The Bayesian trace \ac{mse} of the \ac{mmse} estimator, ${\rm{E}}[\xvec|\yvec]$, is given by (see p. 347 \cite{Kayestimation})
\beqna
\text{BMSE}(\dvec)
\label{conditional_exp}
={\rm{E}}[{\rm{E}}[({\rm{E}}[\xvec|\yvec]- \xvec)^T({\rm{E}}[\xvec|\yvec] - \xvec) | \xvec]].
\eeqna

As discussed in Subsection \ref{subsec;estimator} (after \eqref{x_distribution}), the \ac{gfr}-\ac{ml} estimator from \eqref{estimator}-\eqref{estimator_matrix} is equivalent to the \ac{mmse} estimator in the Bayesian setting, i.e. ${\rm{E}}[\xvec|\yvec]=\hat{\xvec}$,
where $\hat{\xvec}$ is defined in \eqref{estimator} and $ \xvec$ is  treated as a random variable in \eqref{conditional_exp}.
Thus, the inner (conditional) expectation in \eqref{conditional_exp} coincides with the non-Bayesian \ac{mse} expression derived in \eqref{mse_general}, with $\xvec$ now treated as random.  By substituting \eqref{mse_general} in \eqref{conditional_exp} and calculating the expectation \ac{wrt} $\xvec$, one  obtains
\beqna\label{bayesian_mse_derivaition}
\text{BMSE}(\dvec)
= {\rm{E}}_{\xvec}\Big[{\text{tr}}\Big(\Kmat^{-2}(\dvec)
{h}_{\text{M}}(\Lmat)\Dmat\Rmat^{-1}\Dmat{h}_{\text{M}}(\Lmat)\Big)\Big]\nonumber\\
+\mu^2 {\text{tr}} \Big({\rm{E}}[(\xvec_0-\xvec)(\xvec_0-\xvec)^T]{h}_{\text{R}}^{+}(\Lmat)\Kmat^{-2}(\dvec){h}_{\text{R}}^{+}(\Lmat)\Big),
\eeqna
where we used the trace operator definition and its property
${\text{tr}}(\Amat\Bmat) = {\text{tr}}(\Bmat\Amat)$. 
By changing the order of the trace and the expectation operators, substituting ${\rm{E}}[(\xvec - \xvec_0)(\xvec - \xvec_0)^T] = \frac{1}{\mu}(h_{\text{R}}^{+}(\Lmat))^{\dagger}$ (according to \eqref{x_distribution}) in \eqref{bayesian_mse_derivaition}, and applying the pseudo-inverse property $\Amat = \Amat\Amat^{\dagger}\Amat$, we simplify \eqref{bayesian_mse_derivaition} as follows:
\beqna\label{mse_bayesian}
{\text{BMSE}}({\dvec}) 
= {\text{tr}}\Big(\Kmat^{-2}(\dvec)({h}_{\text{M}}(\Lmat)\Dmat\Rmat^{-1}\Dmat{h}_{\text{M}}(\Lmat)+\mu h_{\text{R}}^{+}(\Lmat))\Big)\hspace{-1cm}\nonumber\\={\text{tr}}(\Kmat^{-1}(\dvec)),\hspace{4.1cm}
\eeqna
where the last equality is obtained by substituting \eqref{estimator_matrix}. 
The \ac{bmse} in \eqref{mse_bayesian} coincides with the Bayesian \ac{crb} in this case, which is attainable by the \ac{mmse} here, since the posterior distribution of $\xvec$ is Gaussian.
Consequently, the cost function in \eqref{mse_bayesian} can also be interpreted as a  Bayesian \ac{crb} cost function \cite{6981988}.

\vspace{-0.2cm}
\subsection{Cost Function 4: WC-BMSE}\label{subsec;worstBayesianMSE}
\vspace{-0.05cm}
Instead of taking the trace of the \ac{bmse} matrix (i.e. its  Frobenius norm \cite{Horn2012}, pp. 341-342) as in \eqref{mse_bayesian}, we here consider the spectral norm (\hspace{1sp}\cite{Horn2012}, p.  346) of the \ac{mse} matrix:
\beqna\label{wc_mse_bayesian} {\text{BMSE}{wc}}(\dvec) = \lambda_{\text{max}}\Big(\Kmat^{-1}(\dvec)\Big) = \lambda_{\text{min}}^{-1}\Big(\Kmat(\dvec)\Big), \eeqna
where the last equality follows from Theorem 4.2.2 in \cite{Horn2012}.

It is important to note that if $\Kmat(\dvec)$ is not invertible for some $\dvec$, the inverse is replaced by the pseudo-inverse. Consequently, the second equality holds for the minimal eigenvalue that is nonzero. The reformulation on the r.h.s. of \eqref{wc_mse_bayesian} eliminates the need to invert $\Kmat(\dvec)$, and is therefore more computationally efficient.

 \begin{claim}
     The cost function in \eqref{wc_mse_bayesian} can be interpreted as the \ac{wcbmse} in the Bayesian approach, measured in the Mahalanobis distance sense. 
 \end{claim}
 \begin{proof}
 In this proof we adopt a set-based (worst-case) description of the
uncertainty in the random variable 
$\zvec \define [\evec^T, \xvec^T]^T$. 
First, it can be seen that $\zvec \sim \mathcal{N}(\bar{\zvec}, \Rmat_{\zvec})$, where
$\bar{\zvec}\define[
\zerovec^T,\xvec_0^T]^T$ and 
$\Rmat_{\zvec}\define\begin{pmatrix}
 \Rmat & \zerovec\\\zerovec & (h_{\text{R}}^{+}(\Lmat))^{\dagger}\end{pmatrix}$.
 Let $\zvec_0$ denote an arbitrary realization of $\zvec$. We restrict this realization to the Mahalanobis ellipsoid implied by the prior:
 \beqna
 \label{const2}
\|\zvec_0-\bar{\zvec}\|_{\Rmat_{\zvec}^{\dagger}}\leq 1,
 \eeqna
 and we then maximize the $\ell_2$-norm of the the error vector from \eqref{error_vec} over all $\zvec_{0}$ in this set, as follows
  \beqna\label{wc_bmse_step1}
 \max_{\|\zvec_0-\bar{\zvec}\|_{\Rmat_{\zvec}^{-1}}\leq 1}\|\Kmat^{-1}(\dvec)\Hmat\Rmat_{\zvec}^{\dagger}(\zvec_0-\bar{\zvec})\|^2,
 \eeqna
where $\Hmat\define(h_{\text{M}}(\Lmat), \Imat)$. 
Using the Rayleigh-Ritz Theorem for the vector $(\Rmat_{\zvec}^{1/2})^{\dagger} (\zvec_0 - \bar{\zvec})$, the solution to \eqref{wc_bmse_step1} is expressed as
 \begin{equation}
\label{wc_bmse_step2}
\lambda_{\text{max}}((\Rmat_{\zvec}^{1/2})^{\dagger}\Hmat^T\Kmat^{-2}(\dvec)\Hmat(\Rmat_{\zvec}^{1/2})^{\dagger})=\lambda_{\text{max}}(\Kmat^{-1}(\dvec)),\hspace{-0.1cm}
 \end{equation}
where we substituted $\Kmat(\dvec)$ from \eqref{estimator_matrix} and used the property $\lambda_i(\Amat \Amat^T) = \lambda_i^2(\Amat^T \Amat)$  with $\Amat = \Kmat^{-1}(\dvec, \mu) \Hmat (\Rmat_{\zvec}^{1/2})^{\dagger}$.  It can be seen that the last term in \eqref{wc_bmse_step2} coincides with ${\text{BMSE}_{wc}}(\dvec)$, which completes the proof.
 \end{proof}

\vspace{-0.25cm}
\subsection{Discussion and General Remarks}\label{subsec;remarks}
\vspace{-0.05cm}
The cost functions in \eqref{CRB}, \eqref{worst_mse}, \eqref{mse_bayesian}, and  \eqref{wc_mse_bayesian} are not  functions of the unknown input graph signal, $\xvec$. This property enables their practical use for sampling design by replacing $C(\dvec)$ in \eqref{sampling_statment} with any of these cost functions.
These cost functions can also be utilized for general system design, such as selecting graph filters for regularization.
In the following we discuss some properties and special cases of the proposed approach.
\marksnippet{Extreme_case1}
\subsubsection{Extreme Cases}  
\textcolor{black}{We briefly examine two boundary cases to gain insight into the behavior of the cost functions.}\newline
{\textcolor{black}{\textbf{(i) Full observability without regularization:}  
When all nodes are observed ($\dvec = \onevec$) and no regularization is used ($\mu = 0$),}} \eqref{estimator_matrix} is reduced to 
\beqna
\label{estimator_matrix_d1_mu0}
\Kmat(\dvec=\onevec,\mu=0)
= {h}_{\text{M}}(\Lmat)\Rmat^{-1}{h}_{\text{M}}(\Lmat).
\eeqna
Substituting \eqref{estimator_matrix_d1_mu0} into any of the proposed cost functions, \eqref{CRB}, \eqref{worst_mse}, \eqref{mse_bayesian}, and \eqref{wc_mse_bayesian}, results in 
\beqna
\label{term_extreme}
C(\hat{\xvec},\dvec=\onevec)={\text{tr}}(({h}_{\text{M}}(\Lmat)\Rmat^{-1}{h}_{\text{M}}(\Lmat))^{\dagger}),
\eeqna 
{\textcolor{black}{and the bias $\bvec(\dvec)$ becomes zero. This special case unifies the behavior of all  cost functions and provides a baseline reference.}}


\color{black}\textbf{(ii) Dominant regularization:}  
When $\mu \to \infty$, it 
is shown in \textcolor{black}{Subsection \ref{app;limit_cost_func} in the supplementary material}
that \beqna\label{Extreme_case_large_mu} \hat{\tilde{\xvec}}_{\mathcal{V}\setminus\mathcal{R}} = [\tilde{\xvec}_0]_{\mathcal{V}\setminus\mathcal{R}},\eeqna where $\mathcal{V}\setminus\mathcal{R}$ denotes the set of frequencies for which $h_{\text{R}}^{+}(\lambda_i)\neq0$. That is, in the subspace spanned by the image of $h_{\text{R}}^{+}(\Lmat)$, the estimator relies solely on the prior.
As a result, when $h_{\text{R}}^{+}(\Lmat)$ is full column rank, all proposed cost functions become degenerate (i.e. independent of the sampling set), as detailed in Section \ref{app;limit_cost_func} of the supplementary material. 
This case highlights \color{black} the importance of careful tuning of $\mu$ to balance prior information and observed data. 

\subsubsection{Relation with the Laplacian-Regularized  Design}\label{sc_smoothness_discuss} In \cite{Bai2020},  the following sampling allocation was proposed:
\beqna\label{LR_design}
\dvec^{LR} 
= \arg\max_{\dvec\in\{0,1\}^N:\|\dvec\|_2^2=q} \lambda_{\text{min}}(\Dmat^T \Dmat + \mu \Lmat).
\eeqna 
It can be seen that \textcolor{black}{for the special case presented in Subsection \ref{sc_smooth}, where} $h_{\text{R}}^{+}(\Lmat) = \Lmat$, and $\Lmat\xvec_0 =\zerovec$, if $h_{\text{M}}(\Lmat)=\Imat$, $\Rmat = \Imat$, the \ac{wcbmse} cost function from \eqref{wc_mse_bayesian} coincides with the Laplacian-regularized sampling. 
The proposed \ac{wcbmse} cost function in this work can be seen as a generalization to different choices of $h_{\text{M}}(\Lmat)$, $h_{\text{R}}^{+}(\Lmat)$, $\xvec_0$, and $\Rmat$.

\subsubsection{Relationship with Bandlimitedness-Based Approaches} \label{sc_bandlimited_discuss}
\marksnippet{relation_with_bl1}
\textcolor{black}{Many popular sampling and recovery strategies in \ac{gsp}, such as the widely-used A-design~\cite{anis2016efficient} and E-design~\cite{chen2015discrete}, are based on the assumption that the graph signal is strictly bandlimited. These approaches optimize the performance of estimators constrained to a known subspace and work well under perfect bandlimitedness. The proposed framework generalizes them by incorporating prior knowledge through regularization for various, not necessarily bandlimited, signal models. }
\marksnippet{r1_relation_with_bl21}
In particular, 
the A-design \cite{anis2016efficient} \textcolor{black}{approach, which minimizes the mean \ac{mse},} and the E-design \cite{chen2015discrete} approach, \textcolor{black}{which minimizes the  \ac{wcmse}, both under the constraint of strict bandlimitedness,}  are given by  
\beqna\label{a_design}
\dvec^{A-des.} = \arg\min_{\dvec\in\{0,1\}^N:\|\dvec\|_2^2=q} {\text{tr}}(({\Vmat}_{\mathcal{S},\mathcal{R}}^T\Rmat_{\mathcal{S},\mathcal{S}}^{-1} {\Vmat}_{\mathcal{S},\mathcal{R}})^{-1}),\\
\label{e_design}
\dvec^{E-des.} = \arg\max_{\dvec\in\{0,1\}^N:\|\dvec\|_2^2=q} \lambda_{\text{min}}({\Vmat}_{\mathcal{S},\mathcal{R}}^T\Rmat_{\mathcal{S},\mathcal{S}}^{-1} {\Vmat}_{\mathcal{S},\mathcal{R}}),\hspace{0.3cm}
\eeqna
where $\mathcal{R}\subseteq\mathcal{V}$ is the subset of frequency indices associated with the bandlimited graph signal. \textcolor{black}{This special case is further discussed in Section~\ref{appendix_relation_bandlimited} of the supplementary material.}

The following claim states that for the special case where the measurement model captures a bandlimited graph signal over the frequency set $\mathcal{R}$ and the regularization strongly suppresses 
frequency components in the rest of the spectrum, $\mathcal{V}\setminus\mathcal{R}$, the proposed cost functions align with the  A-design and E-design cost functions. Thus, the proposed methods can be interpreted as generalizations of the  A-design and E-design criteria for general graph filters.
\marksnippet{r1_claim_a_E1}
\begin{claim}
\label{claim_A_E}
\textcolor{black}{Consider the special case of estimating a strictly graph-bandlimited signal in Subsection \ref{special_case_bandlited} } 
with \textcolor{black}{$\mu\rightarrow\infty$. Then, if $h_{\text{M}}(\Lmat)=\Imat$  
the \ac{bcrb} from \eqref{CRB} and } the  \ac{bmse} from \eqref{mse_bayesian} coincide with the A-design cost from \eqref{a_design}. In addition, 
the \ac{wcbmse} from \eqref{wc_mse_bayesian} coincides with the E-design cost function from \eqref{e_design}.
\end{claim}
\begin{IEEEproof}
The proof appears in Appendix \ref{appendix_relation_bandlimited} \textcolor{black}{in the supplementary material.}
\end{IEEEproof}
\vspace{-0.2cm}
\section{Sensor selection solvers}
\label{solvers_subsection}
\vspace{-0.05cm}
Finding the set of $q<N$ sensor locations from the $N$ nodes to minimize the different cost functions, as described in \eqref{sampling_statment},   is a combinatorial optimization problem that has, in the worst
case, a computational complexity of $\binom{N}{q}$, which is prohibitive for large-scale systems.  
Thus, we propose two iterative approaches as follows. 
In Subsection \ref{subsec;greedy}, we present a heuristic method that iteratively selects sensor locations in order to approximate the solution of \eqref{sampling_statment}.
In Subsection \ref{subsec;PGD} we derive a \ac{pgd} method to solve a relaxation of \eqref{sampling_statment},   providing a computationally efficient solution with a complexity reduction by a factor of approximately $N$.
\subsection{Greedy Algorithm}\label{subsec;greedy}
In this subsection, we introduce a greedy algorithm, described in Algorithm \ref{Algorithm_greedy},  for the practical implementation of the sensor selection problem.
While greedy algorithms do not guarantee optimality, they often perform well in practice. The core idea behind this algorithm is to iteratively add to the sampling set those nodes that minimize the chosen cost function. The stopping criterion depends on the chosen optimization: either selecting a fixed number of sensors $q$ (as stated in \eqref{sampling_statment}) or achieving a predefined error threshold (as stated in \eqref{eq_l0}). 

\begin{algorithm}[hbt]
    \textbf{Input:}
      graph filters, $h_{\text{R}}^{+}(\Lmat)$, $h_{\text{M}}(\Lmat)$,  number of nodes, $q$, noise covariance matrix, $\Rmat$, and regularization parameter, $\mu$
      \\
       \textbf{Initialization:} Set the initial sampling subset $\mathcal{S}^{(0)}=\emptyset$ and the iteration index, $i=0$\\
        \textbf{while} {$i<q$}
        \begin{enumerate}
            \item Update the set of available nodes: $\mathcal{L}=\mathcal{V}\setminus\mathcal{S}^{(i)}$
            \item Select the node that minimizes the cost function:
            \begin{equation}\label{objective}
                w^{opt} = \arg \min_{w\in {\mathcal{L}}} C(\onevec_{\{\mathcal{S}^{(i)}\cup w\}})
            \end{equation}
            where $C(\dvec)$ is one of the proposed cost functions defined in \eqref{CRB}, \eqref{worst_mse},  \eqref{mse_bayesian}, or \eqref{wc_mse_bayesian}
            \item Update the sampling set: $\mathcal{S}^{(i+1)} \leftarrow \mathcal{S}^{(i)}\cup w^{opt}$, and increment the iteration index, $i\leftarrow i+1$
        \end{enumerate}
        \textbf{Output:}  Subset of the selected $q$ nodes: $\mathcal{S}=\mathcal{S}^{(i)}$
        \caption{Greedy Selection of the Measured Nodes}\label{Algorithm_greedy}
\end{algorithm}
\subsubsection{Computational Complexity}
\color{black}
The computational complexity of Algorithm \ref{Algorithm_greedy}  depends on the complexity of evaluating the cost function $C(\dvec)$ for each set of sensors (i.e. in each iteration, in \eqref{objective}) and on the desired number of sensors, $q$. Since a single calculation of \( C(\dvec) \) from \eqref{CRB}, \eqref{worst_mse}, and \eqref{mse_bayesian}  has a computational complexity of \( O(N^3) \) (the computational complexity of \eqref{wc_mse_bayesian} is $O(N^2)$ \cite{Lanczos_algorithm_eig}), this method requires \( \sum_{n=0}^{q-1} (N - n) \) calculations of \( C(\dvec) \) without any search-reduction rules (i.e. \( N \) options for the first node, \( N-1 \) options for the second node, and so on). 
 \marksnippet{greedy_complexity1}
This summation results in \( q(N - \frac{q-1}{2}) \) calculations, \textcolor{black}{ which for $q=1$ results in $N$ repetitions, and for $q=N$ results in  $0.5N^2$ repetitions, i.e.  so asymptotically it behaves like $qN$}. 
The total complexity becomes \( O(\textcolor{black}{qN^4}) \) (or $O(\textcolor{black}{qN^3})$ for  \eqref{wc_mse_bayesian}),
 which is computationally prohibitive for a large $N$. 
 Consequently, while the greedy approach is straightforward, its high computational cost for large graphs necessitates the development of more efficient methods. 
 \textcolor{black}{For the following special case,  the complexity can be further reduced. }

\color{black}
\color{black}
\marksnippet{submodularity1}
\subsubsection{Submodularity} 
Submodular and monotone functions offer theoretical guarantees for greedy optimization.
When $\Rmat$ is diagonal and $h_{\text{R}}^{+}(\Lmat)$ is positive definite, it can be shown that the negative \ac{bmse} is both submodular and monotonically increasing (see proof in Section \ref{sec_submodularity_bmse} in the supplementary material), and therefore enjoys the associated performance guarantees.
 \marksnippet{efficient_calc1}
\subsubsection{Efficient Cost Update for Single Node Addition}
Consider adding a new sensor \( i \notin \mathcal{S} \), and assuming a diagonal noise covariance matrix \( \Rmat \). The updated of $\Kmat(\dvec)$ is given by (see Proposition \ref{theorem_modular_K} in the supplementary material)
\[
\Kmat(\onevec_{\mathcal{S} \cup \{i\}}) = \Kmat(\onevec_{\mathcal{S}}) + \rvec_i \rvec_i^T,
\]
where \( \rvec_i = \sqrt{\Rmat^{-1}_{i,i}} [h_{\text{M}}(\Lmat)]_{\mathcal{V},i} \) captures the contribution of the new sensor $i$. This rank-one structure enables efficient updates of matrix inverses, as well as the smallest eigenvalues in a greedy search.

First, to avoid recomputing the matrix inverse \( \Kmat^{-1}(\dvec) \) from scratch at each iteration (as needed in the cost functions \eqref{CRB} and \eqref{mse_bayesian}), we apply the Sherman–Morrison identity (see \cite[Eq. 160]{matrix_cookbook}),  as follows:
\beqna\label{efficientKplus1}
\Kmat^{-1}(\onevec_{\mathcal{S} \cup \{i\}}) = \Kmat^{-1}(\onevec_{\mathcal{S}}) - \frac{\Kmat^{-1}(\onevec_{\mathcal{S}})\rvec_i \rvec_i^T \Kmat^{-1}(\onevec_{\mathcal{S}})}{1 + \rvec_i^T \Kmat^{-1}(\onevec_{\mathcal{S}}) \rvec_i}.
\eeqna
Substituting \eqref{efficientKplus1} into the  \ac{bmse} \eqref{mse_bayesian}, \ac{crb}  \eqref{CRB}, and \ac{wcmse} \eqref{worst_mse}
yields
\vspace{-0.5cm}

%
\beqna\label{biased_crb_effiecent_greedy}
{\text{BMSE}}(\onevec_{\mathcal{S} \cup \{i\}}) 
= \text{tr}\big(\Kmat^{-1}(\onevec_{\mathcal{S}})\big) - \frac{\rvec_i^T \Kmat^{-2}(\onevec_{\mathcal{S}})\rvec_i}{1 + \rvec_i^T \Kmat^{-1}(\onevec_{\mathcal{S}})\rvec_i},\eeqna
\vspace{-0.7cm}
\beqna\label{biased_crb_effiecent_greedy}
{\text{bCRB}}(\onevec_{\mathcal{S} \cup \{i\}})={\text{BMSE}}(\onevec_{\mathcal{S} \cup \{i\}})\hspace{3.1cm}\nonumber\\-\mu\text{tr}\Big(\Big(
\Kmat^{-1}(\onevec_{\mathcal{S}}) - \frac{\Kmat^{-1}(\onevec_{\mathcal{S}})\rvec_i \rvec_i^T \Kmat^{-1}(\onevec_{\mathcal{S}})}{1 + \rvec_i^T \Kmat^{-1}(\onevec_{\mathcal{S}}) \rvec_i}\Big)^{2}h_{\text{R}}^{+}(\Lmat)\Big),
\eeqna
\vspace{-0.5cm}
\begin{align}\label{worst_mse_effiecent_greedy}
\text{MSE}_{WC}(\onevec_{\mathcal{S} \cup \{i\}}) = \text{bCRB}(\onevec_{\mathcal{S} \cup \{i\}})\hspace{3.4cm}
\nonumber\\ +\mu^2 \sigma^2_{\text{max}}\Big(\Big(
\Kmat^{-1}(\onevec_{\mathcal{S}}) - \frac{\Kmat^{-1}(\onevec_{\mathcal{S}})\rvec_i \rvec_i^T \Kmat^{-1}(\onevec_{\mathcal{S}})}{1 + \rvec_i^T \Kmat^{-1}(\onevec_{\mathcal{S}}) \rvec_i}\Big)h_{\text{R}}^{+}(\Lmat)\Big).
\end{align}
These updates reduce the per-iteration complexity from \( \mathcal{O}(N^3) \) to \( \mathcal{O}(N^2) \), resulting in a total complexity of \( \mathcal{O}(qN^3) \) for \( q \) selected nodes in Algorithm~\ref{Algorithm_greedy}.

Second, to efficiently estimate the smallest eigenvalue required in \eqref{worst_mse}, we apply results from perturbation theory.
In particular, 
for a given sampling set \( \mathcal{S} \), let \( \lambda_{\min}(\Kmat(\onevec_{\mathcal{S}})) \) and its eigenvector \( \vvec_{\min} \) be known. Then, the first-order Taylor expansion of the minimum eigenvalue function (see \cite[Eqs. 67 and 488]{matrix_cookbook}, implies that
\begin{equation}\label{perturbation_theory}
\lambda_{\text{min}}\big(\Kmat(\onevec_{\mathcal{S}})+\rvec_i\rvec_i^T\big)
\approx \lambda_{\text{min}}\big(\Kmat(\onevec_{\mathcal{S}})\big) + (\vvec_{\text{min}}^T  \rvec_i)^2.
\end{equation}
Substituting \eqref{perturbation_theory} into \eqref{wc_mse_bayesian},  we approximate the \ac{wcbmse} as
\beqna\label{wc_mse_bayesian_efficient}
{\text{BMSE}{wc}}(\onevec_{\mathcal{S} \cup \{i\}}) 
\approx \left( \lambda_{\text{min}}\big(\Kmat(\onevec_{\mathcal{S}})\big) + (\rvec_i^T \vvec_{\text{min}})^2 \right)^{-1}.
\eeqna
This only requires \( \mathcal{O}(N) \) computation  
given \( \lambda_{\text{min}}\big(\Kmat(\onevec_{\mathcal{S}})\big) \) and \( \vvec_{\text{min}} \). 
Since this operation is repeated for \( q(N - \frac{q-1}{2}) \approx qN \) iterations as discussed after Alg. \ref{Algorithm_greedy}, the total complexity is \( \mathcal{O}(qN^2) \).

To conclude, with these complexity-reduction techniques, the greedy algorithm has a total complexity of \( \mathcal{O}(qN^3) \) for the \ac{bcrb}, \ac{wcmse}, and \ac{bmse}, and \( \mathcal{O}(qN^2) \) for the \ac{wcbmse}. This enables scalable implementation of greedy selection methods. However, when \( \Rmat \) is non-diagonal or  \( q \) is large, these updates become computationally expensive. In such cases, we propose an efficient \ac{pgd} method,  described next.
\color{black}

 \vspace{-0.2cm}
\subsection{Alternating PGD}\label{subsec;PGD}
\vspace{-0.05cm}
A common approach for dealing with binary decision variables, such as $\dvec$, is to relax them to continuous variables, and subsequently, to project the solution onto the feasible set of the original problem \cite{4663892}. 
To simplify the problem in \eqref{sampling_statment}, we relax the non-convex Boolean constraint $\dvec\in\{0,1\}^N$ to the convex box constraint  $\dvec\in[0,1]^N$,  and the norm constraint 
to a ball constraint, i.e. $\|\dvec\|_2^2 = q$ to $\|\dvec\|_2^2 \leq q$. 
The relaxed optimization problem is then formulated as follows:
\begin{align}\label{eq_regularized}
    \hat{\dvec} &= \arg\min_{\dvec \in [0,1]^N,~\|\dvec\|_2^2\leq q}  C(\dvec). 
\end{align}


To implement this approach, 
we derive the associated \ac{pgd} algorithm (see p. 223 in \cite{bertsekas1999nonlinear}), which iteratively combines a gradient descent step and a projection step. First, a gradient descent step with a backtracking linesearch \cite{Boyd_2004} is performed to determine a step size $\rho$ that reduces the cost function, as follows. 
We start from an initial value, $\rho$, which is iteratively reduced until the following condition no longer holds:
\begin{equation}\label{linesearch_step1} C\bigl(\mathcal{P}(\dvec^{(k)} - \rho \nabla C(\dvec^{(k)}))\bigr) >
C\bigl(\mathcal{P}(\dvec^{(k)})\bigr), \end{equation}
where $\nabla C(\dvec^{(k)})$ is the gradient of the cost function \ac{wrt} $\dvec$ evaluated at $\dvec^{(k)}$, and $\mathcal{P}$ denotes the projection operator onto the feasible set of the original problem, \beqna\mathcal{P}(\yvec)\define \mathcal{P}_{\big\{\forall\dvec\in\{0,1\}^N\big|\|\dvec\|_2^2\leq q\big\}}(\yvec). \eeqna 
Once the step size $\rho$ has been  determined, the vector is updated as:
\begin{equation}\label{linesearch_step2}
\dvec^{(k+1)} = \dvec^{(k)} - \rho \nabla C(\dvec^{(k)}),
\end{equation}
where the update skips projection onto the non-convex binary set to avoid zero-gradient points in subsequent iterations.

Next, we perform a projection of the result onto the constrained set of the relaxed problem from \eqref{eq_regularized}, given by
\begin{equation}
{\dvec}^{(k+1)}\gets \mathcal{P}_{\big\{\forall\dvec\in\mathbb{R}^N\big|\|\dvec\|_2^2\leq q,~ \zerovec\leq\dvec\leq\onevec\big\}}({\dvec}^{(k+1)}).
\end{equation}
The  set $\{\dvec\in\mathbb{R}^N\big|\|\dvec\|_2^2\leq q,\zerovec\leq\dvec\leq\onevec\}$ is the intersection of two convex sets. 
Thus, the projection step
 can be simplified by using 
alternating \ac{pgd} \cite{boyd2003alternating}, where the projection alternates sequentially between the two sets as follows:
\begin{enumerate} \item Projection onto the $\ell_2$-norm constraint: 
\beqna\label{eq_norm_proj}
\mathcal{P}_{\big\{\forall\dvec\in\mathbb{R}^N\big|\|\dvec\|_2^2\leq q\big\}}(\yvec)=\left\{\begin{array}{lr}\frac{q}{\|\yvec\|_2^2}\yvec \quad, &\|\yvec\|_2^2\geq q 
\\ \yvec \quad, &{\text{otherwise}} \end{array},\right.
\eeqna
which is the solution of Problem 4.22, p. 197 in \cite{Boyd_2004}.
\item Projection onto the box constraint: 
\beqna\label{eq_box_proj}
\mathcal{P}_{\big\{\forall\dvec\in\mathbb{R}^N\big|\zerovec\leq\dvec\leq\onevec\big\}}(\yvec)=\max\{\min\{\yvec,\onevec\},\zerovec\},
\eeqna
$\forall \yvec\in \mathbb{R}^N$.
This projection is obtained by the solution to the Euclidean projection onto a rectangle (p. 399 in \cite{Boyd_2004}).  

\end{enumerate}

The \ac{pgd} algorithm is summarized in Algorithm \ref{Algorithm_PGD}. 
\begin{algorithm}[hbt]
    \textbf{Input:}  
    \(\dvec^{(0)} \), cost function,  $C(\dvec)$, initial step size, \( \rho_0 \),  linesearch backtracking parameter,  $0<\beta<1$, maximum number of  \ac{pgd} iterations, $M_1$, 
    and tolerance, \( \varepsilon \)\\
    \textbf{Initialization:} $\rho=\rho_0$\\
    \textbf{for} {$k=0,\ldots,M_1$}
         \begin{enumerate}
         \item \big\{\textbf{while} condition \eqref{linesearch_step1} satisfied 
         \textbf{ do:} $\rho\gets\beta\rho$\big\}\\
         Update solution \[\dvec^{(k+1)} \gets \dvec^{(k)}-\rho\nabla C(\dvec^{(k)})\quad \text{(see \eqref{linesearch_step2})}\]
          \item  Project onto the $\ell_2$-norm constraint set: \[\dvec^{(k+1)}\gets
          \frac{q}{\|\dvec^{(k+1)}\|_2^2}\dvec^{(k+1)}
          \quad \text{(see \eqref{eq_norm_proj})}
    \]
        \item
         Project onto the box constraint:
        \[ \dvec^{(k+1)}\gets
        \max\{\min\{\dvec^{(k+1)},\onevec\},\zerovec\} \quad \text{(see \eqref{eq_box_proj})}
         \]
                 \item $k\gets k+1$, and  $\rho\gets\rho_0$
         \item 
          \textbf{if} $\|\dvec^{(k)}-\dvec^{(k-1)}\|_2\leq\varepsilon$
          \textbf{break} 

           \end{enumerate}
    
        \textbf{Output:} Final solution in the constrained set,  $\mathcal{P}(\dvec^{(k)})$

        \caption{Alternating Projected Gradient Descent for \eqref{eq_regularized}}\label{Algorithm_PGD}
\end{algorithm}

 The gradients of the considered cost functions  used in Step 1) of Algorithm \ref{Algorithm_PGD}
are provided in the following claim.
\vspace{-0.1cm}
\begin{claim}
\label{Claim_3}
The gradient of the cost functions from \eqref{CRB}, \eqref{worst_mse}, \eqref{mse_bayesian}, and \eqref{wc_mse_bayesian} \ac{wrt} $\dvec$ can be written as 
\end{claim}
\vspace{-0.5cm}
\begin{equation}
\label{general_gradient_form}
\nabla C(\dvec) = -2 
{\text{diag}}(\Rmat^{-1} \Dmat{h}_{\text{M}}(\Lmat)\Kmat^{-1}(\dvec)\Qmat\Kmat^{-1}(\dvec){h}_{\text{M}}(\Lmat)),
\end{equation}
{\em{where the matrix $\Qmat$ for each cost function is }}
\beqna\label{gradient_bcrb}
\Qmat_{\text{bCRB}} = \Imat-\Kmat^{-1}(\dvec){h}_{\text{R}}^{+}(\Lmat)
-{h}_{\text{R}}^{+}(\Lmat)\Kmat^{-1}(\dvec),\hspace{1.1cm}
\\
\label{gradient_wcmse}
\Qmat_{\text{MSE-wc}}=\Big(\Imat-\Kmat^{-1}(\dvec){h}_{\text{R}}^{+}(\Lmat)\big(\Imat-\uvec_{\text{max}}\uvec_{\text{max}}^{T} {h}_{\text{R}}^{+}(\Lmat)\big)\hspace{0.2cm}\nonumber\\
-\big(\Imat-{h}_{\text{R}}^{+}(\Lmat) \uvec_{\text{max}}\uvec_{\text{max}}^{T}\big){h}_{\text{R}}^{+}(\Lmat)\Kmat^{-1}(\dvec) \Big),
\\
\label{gradient_bmse}
\Qmat_{\text{BMSE}} = \Imat,\hspace{6.2cm}\\\label{gradient_wc_bmse}
\Qmat_{\text{BMSE}_{wc}}=\lambda^{-2}_{\text{min}}(\Kmat(\dvec))\Kmat(\dvec) \uvec_{\text{min}}\uvec_{\text{min}}^{T} \Kmat(\dvec),\hspace{1.5cm}
\eeqna
{\em{and in \eqref{gradient_wcmse} $\uvec_{\text{max}}$ is the normalized eigenvector corresponding to $\lambda_{\text{max}}\left( {h}_{\text{R}}^{+}(\Lmat) \Kmat^{-2}(\dvec) {h}_{\text{R}}^{+}(\Lmat) \right)$, and in \eqref{gradient_wc_bmse} $\uvec_{\text{min}}$ is the  normalized eigenvector corresponding to $\lambda_{\text{min}}\left( \Kmat(\dvec)\right)$.}}
\begin{proof}
The derivations of \eqref{general_gradient_form}-\eqref{gradient_wc_bmse} appear in Appendix \ref{Appendix_A_sec}.
\end{proof}

\subsubsection{Initialization of Algorithm \ref{Algorithm_PGD}}  When $\Rmat=\sigma^2\Imat$, $\sigma\in{\mathbb{R}}$,  \eqref{gradient_bcrb}-\eqref{gradient_wc_bmse} imply that the $n$th entry of the gradient of $C(\dvec)$ vanishes for $d_n=0$.  Consequently, initializing with $d_n^{(0)}=0$ causes the $n$th entry to remain zero throughout the optimization,  potentially preventing convergence of the algorithm to a lower-cost solution. To avoid this, we recommend initializing $\dvec$ as an arbitrary feasible solution in $ (0,1)^{N}$, for any $\Rmat$.  
This approach reduces the sensitivity of the output to initialization.

\color{black}
\marksnippet{convexity1}
\subsubsection{Convexity} 
If \( \Rmat \) is diagonal and \( h_{\text{R}}^{+}(\Lmat) \) is positive definite, the \ac{bmse} and \ac{wcbmse} are convex in \( w_i \triangleq d_i^2 \in[0,1]\),  $i=1,\ldots, N$ (see Section~\ref{convexity_proof} in the supplementary material). The constraints are affine in \( w_i \), making the overall problem convex, and thus \ac{pgd} is expected to converge to the global optimum. 
Nevertheless, practically, \ac{pgd} is known to results in good performance also in non-convex settings. 
\color{black}

  \subsubsection{Computational Complexity of Algorithm \ref{Algorithm_PGD}} The per-iteration cost of Algorithm~\ref{Algorithm_PGD} is dominated by evaluating
  the cost functions in \eqref{CRB}, \eqref{worst_mse}, \eqref{mse_bayesian}, as well as their gradients in \eqref{gradient_bcrb}-\eqref{gradient_bmse}, each with a complexity of $O(N^3)$. For the \ac{wcbmse} in \eqref{wc_mse_bayesian} and its gradient in \eqref{gradient_wc_bmse}, the complexity is lower, at \( \mathcal{O}(N^2) \), using efficient eigenvalue methods~\cite{Lanczos_algorithm_eig}. 
 Let $M_1$ denote the number of \ac{pgd} iterations required for convergence, 
 and $M_2$ denote the maximum number of backtracking linesearch steps per iteration. Consequently, the total worst-case complexity of the algorithm is $O(M_1 M_2 N^3)$ (or $O(M_1 M_2 N^2)$ for the \ac{wcbmse}).
 In practice, convergence typically occurs within several dozen iterations, $M_1$, and linesearch often succeeds on the first trial, reducing the effective complexity to approximately $O( M_1 N^3)$ (or $O( M_1 N^2)$ for the \ac{wcbmse}).

\marksnippet{computational_complexity_discuss_1}

\color{black}
\subsection{Comparison of the Proposed Solvers}
The computational complexities of both the greedy and the \ac{pgd} algorithms depend on the chosen cost function and the structure of the matrices involved.
In the general case, the greedy algorithm requires approximately \( qN \) evaluations of the cost function, each involving a matrix inversion or eigendecomposition with complexity \( \mathcal{O}(N^3) \). This results in an overall complexity of \( \mathcal{O}(qN^4) \) for the \ac{bcrb}, \ac{wcmse}, and \ac{bmse} cases, and \( \mathcal{O}(qN^3) \) for the \ac{wcbmse} case. When \( \Rmat \) is diagonal, efficient rank-one update formulae (see \eqref{efficientKplus1}-\eqref{worst_mse_effiecent_greedy}) can significantly reduce the per-evaluation cost, yielding a total complexity of \( \mathcal{O}(qN^3) \) for the \ac{bcrb}, \ac{wcmse}, and \ac{bmse} cases, and \( \mathcal{O}(qN^2) \) for the \ac{wcbmse} case. In contrast, the \ac{pgd} algorithm requires \( M_1 \) iterations, each dominated by one cost and gradient evaluation. Assuming up to \( M_2 \) backtracking steps per iteration, the worst-case complexity is \( \mathcal{O}(M_1 M_2 N^3) \) for the \ac{bcrb}, \ac{wcmse}, and \ac{bmse} cases, and \( \mathcal{O}(M_1 M_2 N^2) \) for the \ac{wcbmse} case. 
Hence, for the former objectives, \ac{pgd} is more efficient than the greedy method when $M_1 M_2 \leq q$ when $\Rmat$ is diagonal, and  $M_1 M_2 \leq qN$ otherwise. 
Thus, PGD offers a scalable alternative to the greedy approach, especially for large \( q \), dense noise covariance \( \Rmat \), or when convergence is achieved within fewer than \( q \) effective iterations.
\color{black}


\section{Simulations}\label{sec;sim}
In this section, we compare the performance 
of the \ac{gfr}-\ac{ml} from \eqref{estimator}-\eqref{estimator_matrix}  for synthetic data (Subsection \ref{subsec;synthtic_data_sim}),  electrical networks data (Subsection \ref{subsec;power_sim}), and 
road network data (Subsection \ref{subsec;trans_sim}) 
 using the following  sampling policies:
\begin{enumerate}
\item [(i)]  
A-design~\cite{anis2016efficient}---the method presented in \eqref{a_design}, 
\item[(ii)] 
E-design~\cite{chen2015discrete}---the method presented in \eqref{e_design}, 
\item [(iii)] 
LR-design~\cite{Bai2020}---the method presented in  \eqref{LR_design},
\item[(iv)]  \ac{bcrb}---the method presented in \eqref{sampling_statment} with the cost function from \eqref{CRB},
\item[(v)] \ac{wcmse}---the method presented in \eqref{sampling_statment} with the cost function from \eqref{worst_mse},
\item[(vi)]  \ac{bmse}---the method presented in \eqref{sampling_statment} with the cost function from \eqref{mse_bayesian},
\item[(vii)] \ac{wcbmse}---the method presented in \eqref{sampling_statment} with the cost function from \eqref{wc_mse_bayesian}.
\end{enumerate}
Unless otherwise stated, we set $\mathcal{R} = {1,\ldots N/2}$ for both the A-design and the E-design. 
The cost functions in i)-vii)  lead to combinatorial problems, and, thus, are implemented here 
by
Algorithm \ref{Algorithm_greedy}, and in Subsection \ref{subsec;trans_sim} by Algorithm \ref{Algorithm_PGD}, with the corresponding cost function as $C(\dvec)$. 
In all simulations, the performance is evaluated using  $10,000$ Monte-Carlo simulations.  Unless otherwise is stated, $\xvec$ is generated  by using \eqref{x_distribution}, and $\yvec$ follows the measurement model in \eqref{reduced_model_diag}, with $\xvec$,  $\evec\sim\mathcal{N}(\zerovec,0.01\Imat)$, and $\mu=0.1$. This Bayesian setting is used to provide a comprehensive evaluation of the methods' performance, while avoiding reference to single-value results, 
as in a non-Bayesian setting.

\subsection{Synthetic Data}\label{subsec;synthtic_data_sim}
In this subsection, we evaluate the proposed methods  on an Erd\H{o}s-R$\acute{\text{e}}$nyi graph model with
 $50$ nodes and edge probability of $0.1$.
The edge weights are randomly drawn as $W_{k,n}\sim\mathcal{N}(5,1)$ for each edge $(k,n)\in\xi$.

Figures \ref{fig1}.a-\ref{fig1}.c present the \ac{mse} versus the percentage of sampled nodes, $\tilde{q}\define q/N\times100\%$, for the different sampling  methods, and different signal generative models
based on the model in \eqref{reduced_model_diag} and \eqref{x_distribution} with the graph filters that appear in Table \ref{tabFilters}.  

\begin{table}[hbt]
\begin{center}
\begin{tabular}{|c|c|c|c|}
\hline
 \textbf{Fig.} & \textbf{\ref{fig1}.a}& \textbf{\ref{fig1}.b}& \textbf{\ref{fig1}.c} \\
\hline
$h_{\text{M}}(\Lmat)$  &$(h_{\text{GMRF}}^2)^{\dagger}=\Lmat$
&$\Imat$
&$h_\text{Diff}^{-1}=e^{0.5\Lmat}$\\
\hline
$(h_{\text{R}}^{+}(\Lmat))^{\dagger}$  &$h_{\text{GMRF}}^2=\Lmat^{\dagger}$
&$h_{\text{Tikh}}^2=(\Imat+0.2\Lmat)^{-2}$
&$\Imat$ \\
\hline
\end{tabular}
\caption{Graph filters used to generate the data, where  
$h_{\text{GMRF}}(\Lmat)$ - \ac{gmrf}~\cite{Dong_Vandergheynst_2016,Vassilis_2016}; 
$h_{\text{Tikh}}(\Lmat)$ - Laplacian (Tikhonov) regularization~\cite{Vassilis_2016,isufi2022graph,MITPress}; 
$h_{\text{Diff}}(\Lmat)$ - heat diffusion kernel~\cite{thanou2016learning,Vassilis_2016}.}

\label{tabFilters}
\end{center}
\end{table}
 \vspace{-0.35cm}
In Figs. \ref{fig1}.a-\ref{fig1}.c the \ac{mse} decreases as the number of sampled nodes, $\tilde{q}$, increases, as expected, since there is more information for the estimation approach. 
In Fig. \ref{fig1}.a, we consider a realistic scenario where the estimated graph signal is smooth, but the measured graph signal is not (see, e.g. the power system example in Subsection \ref{subsec;power_sim}). Here, the estimation performance based on the sampling designs by the \ac{bmse} and the \ac{wcbmse} criteria outperforms other sampling methods for any $\tilde{q}$. This is because methods such as A-design, E-design, and LR-design are based on the assumption that the measured signal is smooth/bandlimited, which does not hold here. In addition, the \ac{wcmse}, which combines the \ac{bcrb} and the worst-case bias (see in \eqref{worst_mse}), outperforms the \ac{bcrb}
for $\tilde{q} > 60\%$, highlighting the importance of accounting for bias.

In Fig. \ref{fig1}.b, we consider the case where both $\xvec$ and $\yvec$ are smooth. Here,  the \ac{bmse}- and \ac{wcbmse}-based estimations outperform the estimations based on the A-design, E-design, and LR-design.  Although the data is smooth, it is not strictly bandlimited,  as these methods assume, or have mismatches in graph filters (e.g. the LR-design assumes different $h_{\text{M}}(\Lmat)$ and $h_{\text{R}}^{+}(\Lmat)$ in \eqref{reduced_model_diag} and \eqref{smothness_full_theta} from those used). The \ac{bcrb} and \ac{wcmse} perform poorly on average, as they optimize locally for specific values of $\xvec$: \ac{bcrb} assumes $\xvec = \zerovec$, while \ac{wcmse} assumes $\xvec$ aligns with the largest eigenvector of the bias term in \eqref{mse_general}.

In Fig. \ref{fig1}.c, we consider the test case of an inverse diffusion (high-pass filtering) process, i.e. $h_{\text{M}}(\Lmat)=h_{\text{Diff.}}^{-1}$, which can model anomalies or source localization in different applications
\cite{Sandryhaila2014,drayer2020detection,7528230,9909689}.  
We generate $\xvec$ using all-pass graph filter, $h_{\text{R}}^{+}(\Lmat)=\Imat$.
We set $\mathcal{R} = {N/2, \ldots, N}$ in \eqref{a_design} and \eqref{e_design}, as the measured signal passes through a \ac{hpf}.
The estimation based on   \ac{wcmse}, and \ac{wcbmse} outperform all other methods for $\tilde{q} \leq 80\%$, as their conservative worst-case approach is well-suited to the highly varied data with high-frequency components resulting from the combination of all-pass and high-pass filtering.
In contrast, the other methods
exhibit suboptimal performance due to their model mismatches. Specifically, the A-design and E-design assume the graph filters from \eqref{special_case}, the LR-design assumes ${h}_{\text{R}}^{+}(\Lmat) = \Lmat$, the \ac{bcrb} assumes $\xvec = \zerovec$, and the \ac{bmse}, which averages over all the eigenvalues of the Bayesian matrix in \eqref{mse_bayesian}, is more sensitive to outliers from very high graph frequencies.

\begin{figure*}[hbt]
\captionsetup[subfigure]{labelformat=empty}
     \centering
     \begin{subfigure}[b]{0.3\textwidth}
         \centering
         \includegraphics[width=\textwidth]{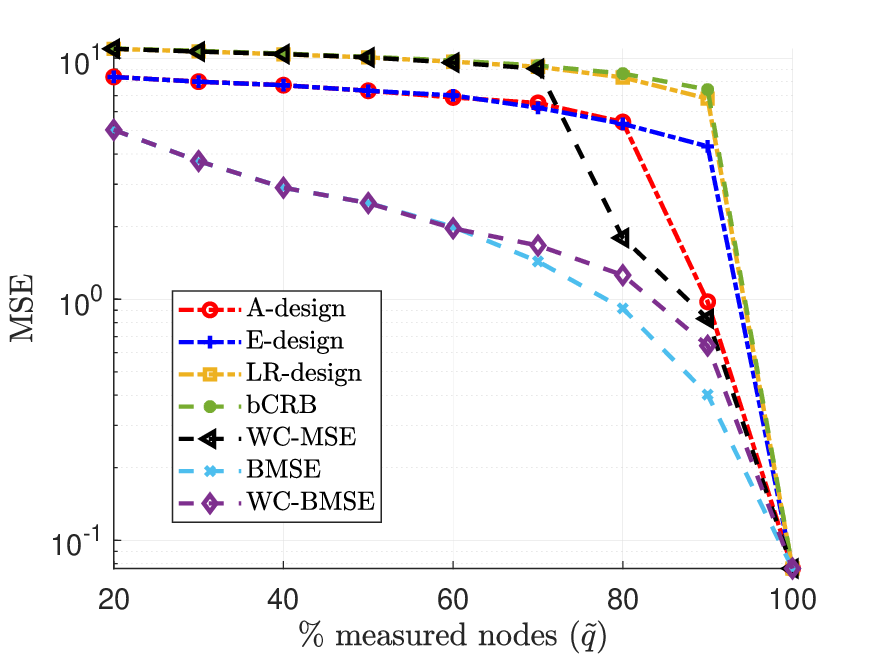}
         \caption{(a)}
         \label{case1}
     \end{subfigure}
     \hfill
     \begin{subfigure}[b]{0.3\textwidth}
         \centering
         \includegraphics[width=\textwidth]{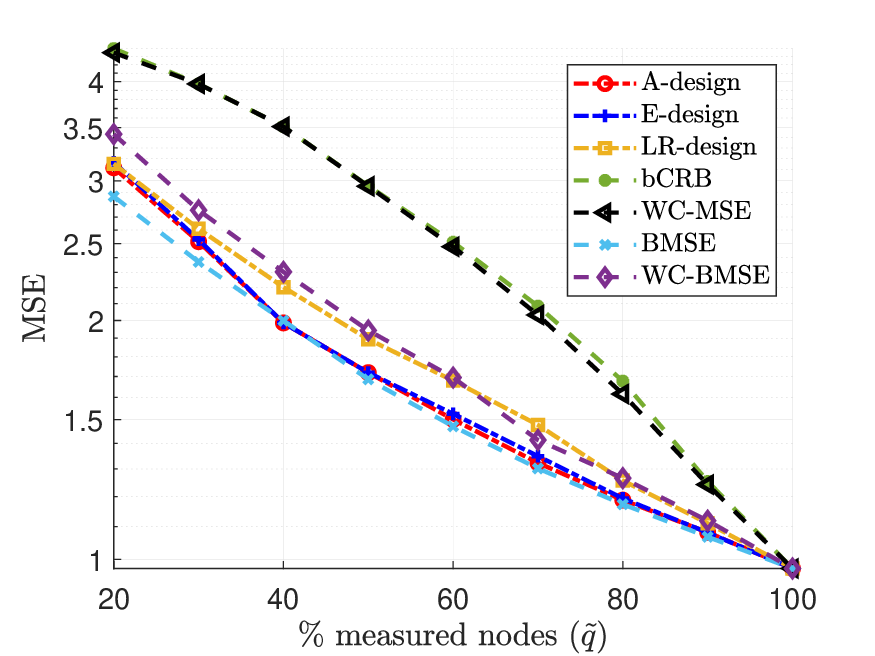}
         \caption{(b)}
         \label{case2}
     \end{subfigure}
     \hfill
     \begin{subfigure}[b]{0.3\textwidth}
         \centering
         \includegraphics[width=\textwidth]{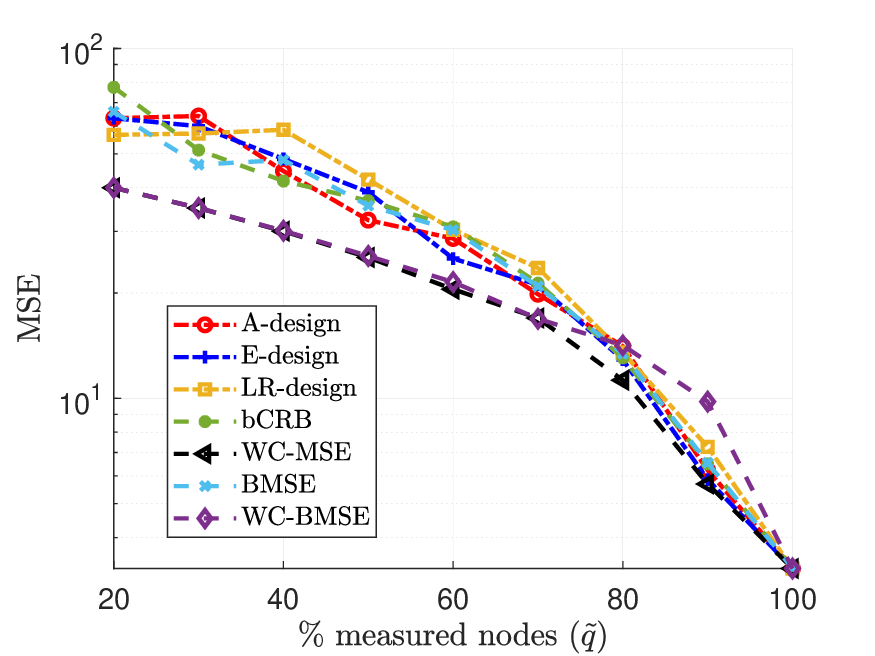}
         \caption{(c)}
         \label{case3}
     \end{subfigure}
     \vspace{-0.25cm}
        \caption{The \ac{mse} versus the percentage of sampled nodes, $\tilde{q}$, of the different sampling sets, 
        where the graph filters are  (see Table \ref{tabFilters}): (a) $h_{\text{M}}(\Lmat)=\Lmat,~(h_{\text{R}}^{+}(\Lmat))^{\dagger}=h_{\text{GMRF}}$; (b) $h_{\text{M}}(\Lmat)=\Imat,~(h_{\text{R}}^{+}(\Lmat))^{\dagger}=h_{\text{Tikh}}$, $\alpha=0.2$; and (3) $h_{\text{M}}(\Lmat)=h_{\text{Diff.}},~\tau=0.5,~h_{\text{R}}^{+}(\Lmat)=\Imat$.      \vspace{-0.25cm}}
        \label{fig1}
\end{figure*}
\subsection{Power System Data}\label{subsec;power_sim}
We now evaluate our sampling strategies in a realistic scenario of sensor allocation for \ac{psse}, a core task in energy management systems. 
A power system can be represented as an undirected weighted graph, ${\mathcal{G}}({\mathcal{V}},\xi)$, where nodes represent buses (generators or loads) and edges represent transmission lines  \cite{2021Anna,drayer2020detection,Giannakis_Wollenberg_2013}.
The graph Laplacian matrix, $\Lmat$, is constructed using line susceptances \cite{dabush2023state,drayer2020detection},
and the grid dynamics are described by 
nonlinear power flow equations, which are often linearized DC  models \cite{Giannakis_Wollenberg_2013}.

This linear model can be written as \eqref{original_model} with $h_{\text{M}}(\Lmat)=\Lmat$, where $\yvec$ is the active power vector and  $\xvec$ is the unknown system state vector, both can be treated as graph signals \cite{drayer2020detection,dabush2023state}.
\ac{psse}  aims to estimate $\xvec$ based on system measurements $\yvec$.
The voltage data of power grids
has been shown empirically and theoretically to be smooth/a  \ac{lpf} signal \cite{dabush2023state,2021Anna,drayer2020detection}. Thus, we model it by using $h_{\text{R}}^{+}(\Lmat)=\Lmat$.
To ensure identifiability,  bus 111 is set to be a reference bus with zero phase.
In the following, we use power data (susceptances and voltage angles) obtained from the IEEE 118-bus test case~\cite{iEEEdata}, with noise   
covariance $\Rmat=0.01\Imat$.
\subsubsection{Robustness to Sampling Set Size}\label{subsec;power_q}
 Figure~\ref{fig2}.a presents the \ac{mse} of the \ac{gfr}-\ac{ml} estimator versus the percentage of measured nodes, $\tilde{q}$, for the different sampling methods. The sharp decrease in \ac{mse} as $\tilde{q}$ increases from $91\%$ to $100\%$ indicates the phase transition from underdetermined to full observability.
 It can be seen that for $60\%\leq \tilde{q}\leq 90\%$, all proposed methods outperform the A-design, E-design, and LR-design. 
This is because the latter methods rely on the assumption that the measured signal is smooth or bandlimited, which is not satisfied here.  Interestingly, the \ac{wcmse} and \ac{bcrb} achieve better performance than in Fig. \ref{case1}, despite using the same graph filters, 
due to the non-Bayesian setup here: 
a single fixed voltage angle vector is estimated, and is relatively close to the specific voltage angles assumed by these cost functions.

\subsubsection{Robustness to Noise}\label{subsec;power_snr}
In Figure \ref{fig2}.b the \ac{mse} of the \ac{gfr}-\ac{ml} is presented  versus $\frac{1}{\sigma^2}$ for $\tilde{q}=70\%$ for the different sampling methods. 
 It can be seen that all methods are consistent, where the \ac{mse} decreases 
 as $1/\sigma^2$ increases.
The \ac{bcrb} and the LR-design are more robust for small values of $1/\sigma^2$, whereas the \ac{wcbmse} achieves a lower \ac{mse} for high values of $1/\sigma^2$. This highlights a tradeoff between robustness to noise and the alignment of the sampling method with the data characteristics. Here, the noise does not contain graphical information, as $\Rmat = \Imat$, which explains why the \ac{bcrb}, with its reduced reliance on graphical knowledge (assuming $\xvec_0 = \zerovec$), performs well in low \ac{snr} scenarios. However, at high $1/\sigma^2$ (high \ac{snr}), the fit of the sampling policy to the data becomes critical. Other sampling methods, apart from the \ac{wcbmse}, suffer from mismatch assumptions, such as the smoothness or bandlimited nature of the measured signal, 
or an inappropriate choice of $\xvec_0$ in \ac{bcrb} and \ac{wcmse}. Additionally, the \ac{bmse}, which averages over all eigenvalues of the Bayesian matrix in \eqref{mse_bayesian}, is more sensitive to outliers at the high graph frequencies in comparison to the \ac{wcbmse}, which takes the largest eigenvalue.

\subsubsection{Robustness to Topology Mismatches} 
In Fig. \ref{fig2}.c, the \ac{mse} of the \ac{gfr}-\ac{ml} is presented 
 versus the number of missing/added edges, $\Delta \Lmat$, for $\tilde{q}=70\%$.    
To investigate the robustness of the methods to topology perturbations, 
$\Delta \Lmat$ edges were randomly removed or added from the Laplacian matrix, which results in a perturbed Laplacian matrix in the sampling methods and in the \ac{gfr}-\ac{ml} estimator. 
It can be seen that the \ac{bcrb}, with its reduced reliance on graphical knowledge (assuming $\xvec_0 = \zerovec$), is the most robust to topology mismatches. 
In contrast, the \ac{mse} for all the other methods that rely on accurate graph information increase as $\Delta\Lmat$ increases.

\begin{figure*}[hbt]
\captionsetup[subfigure]{labelformat=empty}
     \centering
     \begin{subfigure}[b]{0.3\textwidth}
         \centering
         \includegraphics[width=\textwidth]{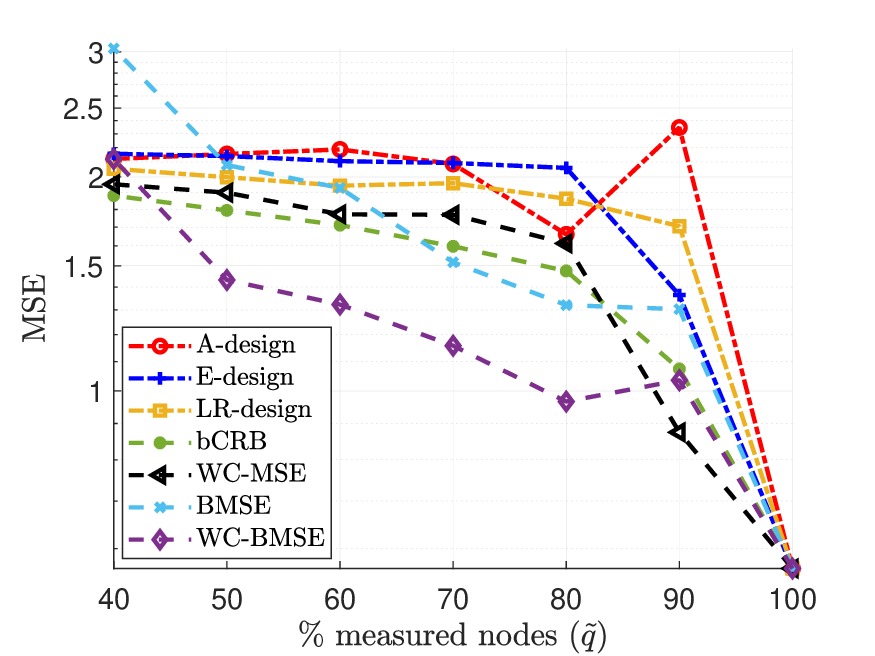}
         \caption{(a)}
         \label{vs_sigma}
     \end{subfigure}
     \hfill
     \begin{subfigure}[b]{0.3\textwidth}
         \centering
         \includegraphics[width=\textwidth]{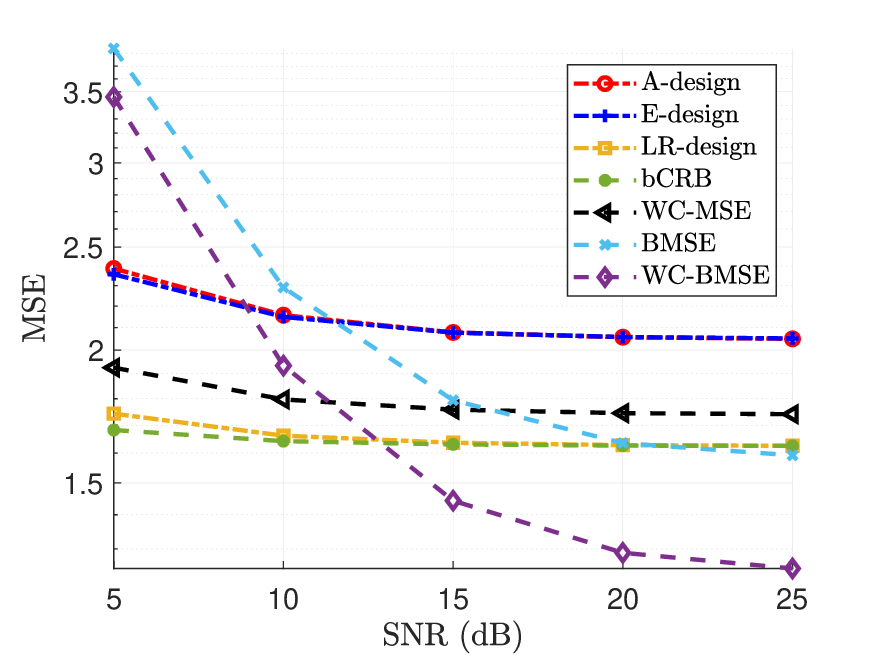}
         \caption{(b)}
         \label{vs_r}
     \end{subfigure}
     \hfill
     \begin{subfigure}[b]{0.3\textwidth}
         \centering
         \includegraphics[width=\textwidth]{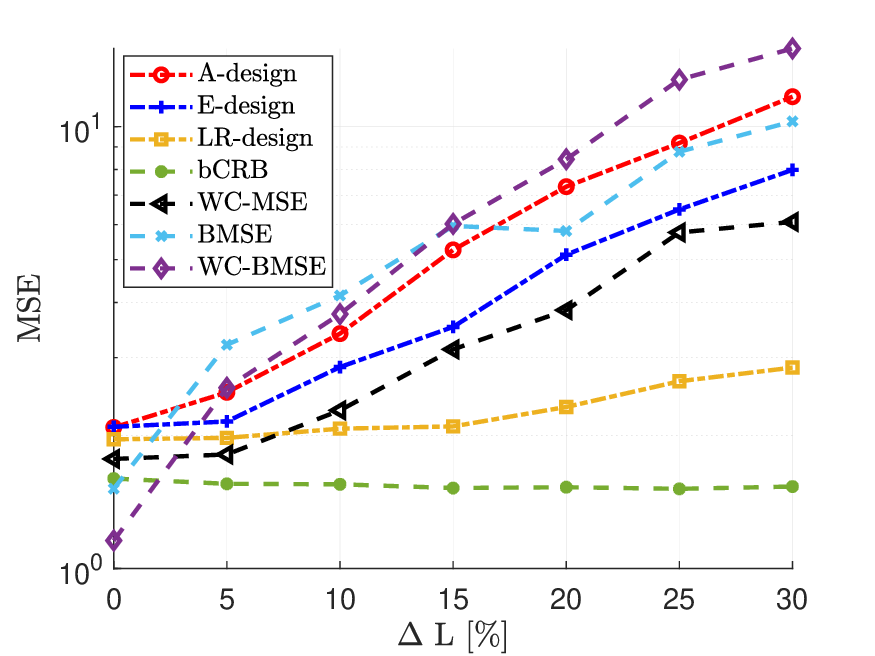}
         \caption{(c)}
         \label{vs_x}
     \end{subfigure}
     \vspace{-0.25cm}
        \caption{State estimation in power systems: the \ac{mse} of the \ac{gfr}-\ac{ml} for all the sampling methods 
    versus (\textbf{a}) the number of measured sensors, $\tilde{q}$, with $\sigma^2=0.01$;  (\textbf{b}) $\frac{1}{\sigma^2}$ with $\tilde{q}=70\%$; and (\textbf{c}) the number of edges that are different from the edges in the true graph, $\Delta \Lmat$.
    \vspace{-0.5cm}}
        \label{fig2}
\end{figure*}
\vspace{-0.25cm}
\subsection{Road Networks Data}\label{subsec;trans_sim}
Road networks can be modeled as graphs, where nodes represent key traffic locations (e.g. intersections), and edges represent roads connecting them \cite{Road_Network_Data_Model2}. In this framework, various traffic metrics, such as vehicle speed, vehicle density, and transit demand, can be treated as graph signals. 
 Due to the large scale of road networks, sensor deployment and data processing become computationally challenging. 
 Thus, sampling a subset of nodes provides a cost-effective solution for monitoring, particularly as traffic flow often exhibits diffusion-like behavior \cite{traffic_diffusion}.

In this section, we evaluate sampling allocation schemes that optimize the estimation performance in road networks. 
Simulations are performed on the Minnesota road graph ($N=2,642$ nodes) from the \ac{gsp} Toolbox \cite{perraudin2014gspbox}, where smooth graph signals are generated via \eqref{reduced_model_diag} using $h_{\text{M}}(\Lmat)=h_{Diff.}$ with $\tau=0.5$, and $h_{\text{R}}^{+}(\Lmat)=h_{Tikh.}$ with $\alpha=0.01$ (see Table 1). 
We compare the \ac{mse} of the \ac{gfr}-\ac{ml} estimator versus the number of nodes, $\tilde{q}$, for all sampling methods.
The methods are implemented using the computationally-efficient \ac{pgd} algorithm from Algorithm \ref{Algorithm_PGD}. Thus, the gradient of the cost function must be computed: 
for the A-design method \cite{anis2016efficient}, the implementation leverages the connection discussed in Claim  \ref{claim_A_E}, utilizing the gradient of the \ac{bmse} provided in \eqref{gradient_bmse}; for the E-design and LR-design methods, the gradient is given by \eqref{gradient_wc_bmse} with $h_{\text{M}}(\Lmat)$ and $h_{\text{R}}^{+}(\Lmat)$ from \eqref{special_case}, with $\mu=10^4$ and $\mathcal{R}=\{1,\ldots,N/2\}$,  for the E-design \cite{chen2015discrete}, and $h_{\text{M}}(\Lmat)=\Imat$ and $h_{\text{R}}^{+}(\Lmat)=\Lmat$  for the LR-design. These gradients are a byproduct of our work.

Figure~\ref{fig:road_net} shows the \ac{mse} of the \ac{gfr}-\ac{ml} estimator versus sampling ratio $\tilde{q}$.  As expected, the \ac{mse} decreases with more samples for all methods. 
For $q = 90\%$, all sampling methods yield similar results due to near-complete observability. 
For $q = 20\%$,  all methods yield comparable (poor) results, due to insufficient data for accurate estimation.   In the intermediate range of $40\%\leq \tilde{q} \leq 80\%$, the \ac{bmse} criterion outperforms the estimation based on the A-design, E-design, and LR-design 
since the data is smooth, but not perfectly bandlimited. 
In addition, the \ac{bmse}-based sampling outperforms the \ac{bcrb}, \ac{wcmse}, and \ac{wcbmse} methods, because it balances flexibility and generality. Unlike the \ac{bcrb} and \ac{wcmse}, which optimize for specific values of $\xvec$, and the overly conservative \ac{wcbmse}, the \ac{bmse} leverages a Bayesian prior to perform well across a wide range of $\xvec$, ensuring better overall performance.
  \begin{figure}[hbt]
    \centering
    \includegraphics[width=0.7\columnwidth]{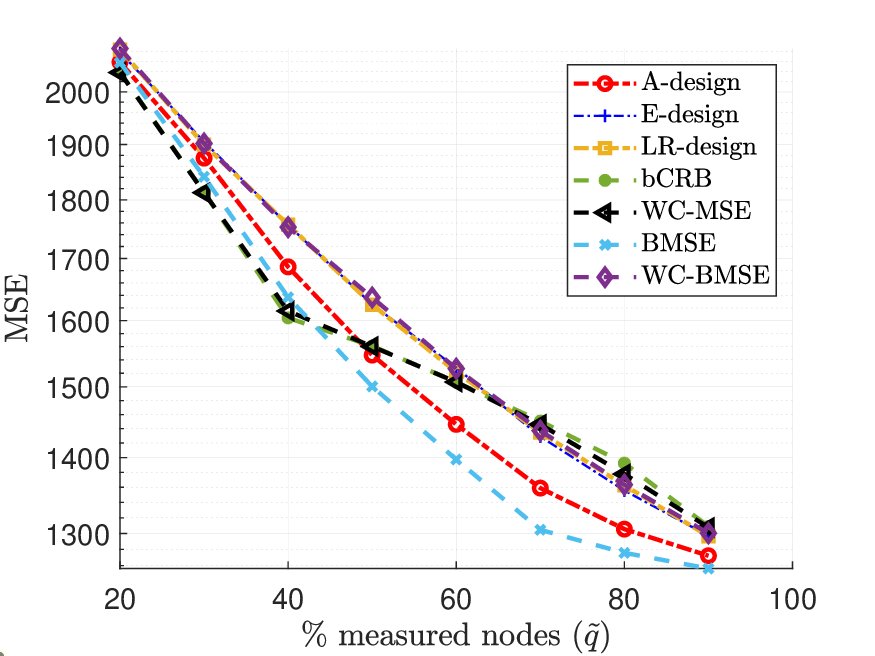}
    \caption{Estimation in  Minnesota road graph: The \ac{mse} versus the percentage of sampled nodes, $\tilde{q}$, of the different sampling sets, for  $h_{\text{M}}(\Lmat)=h_{Diff.}$,  $\tau=0.5$, and $h_{\text{R}}^{+}(\Lmat)=h_{Tikh.}$,  $\alpha=0.01$. \vspace{-0.5cm}}\label{fig:road_net}
\end{figure}
	\section{Conclusion}
	\label{Conclusions}
In this paper, we propose a general framework for efficient sampling allocation for graph signal recovery using the \ac{gfr}-\ac{ml} estimator. By utilizing general graph filters for both the measurement model and the regularization, the \ac{gfr}-\ac{ml} approach effectively enhances recovery performance in underdetermined systems.
We show that its \ac{mse} depends on the unknown parameter and, thus, cannot be used as a criterion for optimization {\textcolor{black}{in the general case}}. Therefore, we introduce four task-specific cost functions to optimize sampling allocation strategies: \ac{bcrb}, \ac{wcmse}, \ac{bmse}, and \ac{wcbmse}. We investigate their properties and establish connections to existing methods, including A-design, E-design, and the LR approach. To address computational complexity, we developed a greedy algorithm and an alternating \ac{pgd} method. 

Simulation results on synthetic graphs, the IEEE 118-bus power system, and the Minnesota road network demonstrate significant \ac{mse} reductions compared to existing methods.
\marksnippet{r1_relation_with_bl_conc1}
{\textcolor{black}{In particular, our framework consistently outperforms the A-design and E-design schemes in terms of \ac{mse}, as these designs are based on the mismatched assumption of perfectly bandlimited signals. }}
For the tested cases, the \ac{wcmse} demonstrates the best fit for high-frequency data, the \ac{wcbmse} and the \ac{bmse} achieve the best performance in most Bayesian settings, and the \ac{bcrb}  exhibits high robustness to noise and topology mismatches.
Our framework demonstrated flexibility and scalability, accommodating general graph filters and enabling task-specific adaptations.

\appendices
\renewcommand{\thesectiondis}[2]{\Alph{section}:}

\vspace{-0.1cm}
 \section{\textcolor{black}{Proof of Claim \ref{Claim_3}}}
\label{Appendix_A_sec}
\vspace{-0.05cm}
 In this appendix, we derive the gradient of the cost functions from \eqref{CRB}, \eqref{worst_mse}, \eqref{mse_bayesian}, and \eqref{wc_mse_bayesian}. For simplicity, 
 \( \Kmat(\dvec) \) from \eqref{estimator_matrix}, ${h}_{\text{M}}(\Lmat)$, and ${h}_{\text{R}}^{+}(\Lmat)$ are denoted as \( \Kmat \), ${{h}_{\text{M}}}$, and ${{h}_{\text{R}}^{+}}$, respectively, throughout the derivations. The order of the following subsections is structured to simplify the mathematical steps, rather than reflecting their order of use in the paper.
\vspace{-0.25cm}
 \subsection{Derivation of the Gradient of  
 \eqref{mse_bayesian}
}\label{gradient_BCRB}
In order to derive the gradient of \eqref{mse_bayesian} \ac{wrt} the vector $\dvec$, we note that by using the linearity of derivatives, the $n$th component of the gradient of   the \ac{bmse} \ac{wrt} $\dvec$ is
\beqna\label{gradient_bmse_n_entry}
[\nabla \text{BMSE}(\dvec)]_n =\frac{\partial}{\partial d_n} {\text{tr}}(\Kmat^{-1})= {\text{tr}}\left(\frac{\partial\Kmat^{-1}}{\partial d_n}\right),
\eeqna
$n=1,\ldots,N$, where the second equality is a known derivative formula (see Eq. (36) in \cite{matrix_cookbook}). 
From (40) in \cite{matrix_cookbook}, we have
\beqna\label{derivative_inv_K_step1}
\frac{\partial\Kmat^{-1}}{\partial d_n}=-\Kmat^{-1} \frac{\partial \Kmat}{\partial d_n} \Kmat^{-1}, ~n=1,\ldots,N. 
\eeqna
Let $\Emat_n$ be a diagonal matrix with $n$ at position $(n,n)$ and zeros elsewhere. The derivative of $\Kmat$ from \eqref{estimator_matrix} is (see (37) in \cite{matrix_cookbook})
\begin{equation}\label{derivative_K}
\frac{\partial \Kmat}{\partial d_n} = {{h}_{\text{M}}} \left( \Emat_n \Rmat^{-1} \Dmat + \Dmat \Rmat^{-1} \Emat_n \right) {{h}_{\text{M}}}.
\end{equation}
By substituting \eqref{derivative_inv_K_step1} and \eqref{derivative_K}  in \eqref{gradient_bmse_n_entry} and using the identity ${\text{tr}}(\Amat)={\text{tr}}(\Amat^T)$, we obtain:
\beqna\label{gradient_BMSE_step4}
\frac{\partial\Kmat^{-1}}{\partial d_n} = -2{\text{tr}}\Big( \Kmat^{-1} {{h}_{\text{M}}}  \Emat_n \Rmat^{-1} \Dmat {{h}_{\text{M}}} \Kmat^{-1}\Big)\nonumber\\=-2{\text{tr}}\Big(\Emat_n \Rmat^{-1} \Dmat {{h}_{\text{M}}}\Kmat^{-2} {{h}_{\text{M}}}\Big),\hspace{0.6cm}
\eeqna
where the last equality follows from ${\text{tr}}(\Amat\Bmat)={\text{tr}}(\Bmat\Amat)$. Using ${\text{tr}}(\Emat_n\Amat)=\Amat_{n,n}$, yields \eqref{general_gradient_form} with $\Qmat$ given by \eqref{gradient_bmse}.
\subsection{Derivation of of the Gradient of \eqref{CRB}}\label{gradient_CRB}
To compute the gradient of  \eqref{CRB}, we substitute the connection that is obtained by changing variables in \eqref{estimator_matrix}, ${{h}_{\text{M}}}\Dmat\Rmat^{-1}\Dmat{{h}_{\text{M}}} = \Kmat-{{h}_{\text{R}}^{+}}$, in \eqref{CRB}: 
\begin{equation}\label{bcrb_form2}
\text{BCRB}(\dvec) = {\text{tr}}(\Kmat^{-2}(\Kmat-{{h}_{\text{R}}^{+}}))
=\text{BMSE}(\dvec)-{\text{tr}}(\Kmat^{-2}{{h}_{\text{R}}^{+}}).  \nonumber 
\end{equation}
The gradient of this \ac{bcrb} is 
\begin{equation}\label{gradient_bcrb_1_step}
\nabla \text{BCRB}(\dvec) = \nabla\text{BMSE}(\dvec)-\nabla{\text{tr}}(\Kmat^{-2}{{h}_{\text{R}}^{+}} ).  
\end{equation}
The $n$th component of the last term of \eqref{gradient_bcrb_1_step} is
\beqna\label{gradient_bcrb_n_entry}
[\nabla {\text{tr}}(\Kmat^{-2}{{h}_{\text{R}}^{+}} )]_n=\frac{\partial}{\partial d_n} {\text{tr}}(\Kmat^{-2}{{h}_{\text{R}}^{+}})
= {\text{tr}}\left(\frac{\partial\Kmat^{-2}}{\partial d_n}{{h}_{\text{R}}^{+}}\right),\hspace{0cm}
\eeqna
where the second equality is a known derivative formula (see Eq. (36) in \cite{matrix_cookbook}), and since $h_{\text{R}}^{+}(\Lmat)$ does not depend on $\dvec$.  
By using derivative rules (see (37) from \cite{matrix_cookbook}), and substituting in \eqref{derivative_inv_K_step1}
\beqna\label{derivative_square_inv_K}
\frac{\partial\Kmat^{-2}}{\partial d_n}=-\Kmat^{-2}\frac{\partial\Kmat}{\partial d_n}\Kmat^{-1}-\Kmat^{-1}\frac{\partial\Kmat}{\partial d_n}\Kmat^{-2}.
\eeqna
By substituting \eqref{derivative_square_inv_K} in \eqref{gradient_bcrb_n_entry} and using the property ${\text{tr}}(\Amat\Bmat)={\text{tr}}(\Bmat\Amat)$ and the property ${\text{tr}}(\Amat)={\text{tr}}(\Amat^T)$ we obtain
\beqna\label{gradient_bcrb_n_entry_step2}
[\nabla {\text{tr}}\Big(\Kmat^{-2}{{h}_{\text{R}}^{+}} \Big)]_n=-2 {\text{tr}}(\Kmat^{-2}\frac{\partial\Kmat}{\partial d_n}\Kmat^{-1}{{h}_{\text{R}}^{+}}).
\eeqna
By substituting \eqref{derivative_K} in \eqref{gradient_bcrb_n_entry_step2} and using the property ${\text{tr}}(\Amat\Bmat)={\text{tr}}(\Bmat\Amat)$ and the property ${\text{tr}}(\Amat)={\text{tr}}(\Amat^T)$ we obtain 
\beqna\label{gradient_bcrb_n_entry_final}
[\nabla {\text{tr}}\Big(\Kmat^{-2}{{h}_{\text{R}}^{+}} \Big)]_n\hspace{5.cm}\nonumber\\=-2 
{\text{tr}}\Big(\Emat_n\Rmat^{-1} \Dmat{{h}_{\text{M}}}\Kmat^{-1}\big(\Kmat^{-1}{{h}_{\text{R}}^{+}}+{{h}_{\text{R}}^{+}}\Kmat^{-1}\big)\Kmat^{-1}{{h}_{\text{M}}}\Big).
\eeqna
By noting that ${\text{tr}}(\Emat_n\Amat)=\Amat_{n,n}$ we obtain
\beqna\label{gradient_bcrb_final}
\nabla {\text{tr}}(\Kmat^{-2}{{h}_{\text{R}}^{+}} )\hspace{6cm}\nonumber\\=-2 
{\text{diag}}\Big(\Rmat^{-1} \Dmat{{h}_{\text{M}}}\Kmat^{-1}(\Kmat^{-1}{{h}_{\text{R}}^{+}}+{{h}_{\text{R}}^{+}}\Kmat^{-1})\Kmat^{-1}{{h}_{\text{M}}}\Big).
\eeqna
Substituting \eqref{gradient_bmse} into \eqref{general_gradient_form} to obtain $\nabla \text{BMSE}(\dvec)$, and then substituting it and \eqref{gradient_bcrb_final} into \eqref{gradient_bcrb_1_step}, yields \eqref{general_gradient_form} with $\Qmat$ from \eqref{gradient_bcrb}.

\vspace{-0.2cm}
\subsection{Derivation of  the Gradient of \eqref{worst_mse}}\label{gradient_worst_mse}
\vspace{-0.05cm}
In the following we derive the gradient of \eqref{worst_mse} \ac{wrt} $\dvec$ 
First, we note that according to \eqref{worst_mse}, we have
\beqna \label{gradient_wcmse_1step}
\nabla\text{MSE}_{wc}(\dvec) = \nabla\text{bCRB}(\dvec)+\nabla\lambda_{\text{max}}\left( {{h}_{\text{R}}^{+}} \Kmat^{-2} {{h}_{\text{R}}^{+}} \right).
\eeqna
The eigenvalue $\lambda_{\text{max}}\left( {{h}_{\text{R}}^{+}} \Kmat^{-2} {{h}_{\text{R}}^{+}} \right)$ depends on $\dvec$ through $\Kmat$. The gradient of the maximal eigenvalue \ac{wrt} a parameter is given by (see, e.g. Eq. (67) in \cite{matrix_cookbook})
\beqna\label{gradient_wmse_step1}
\frac{\partial \lambda_{\text{max}}\left( {{h}_{\text{R}}^{+}} \Kmat^{-2} {{h}_{\text{R}}^{+}} \right)}{\partial d_n}
= \uvec_{\text{max}}^{T} {{h}_{\text{R}}^{+}}\frac{\partial \left(  \Kmat^{-2}\right)}{\partial d_n}  {{h}_{\text{R}}^{+}} \uvec_{\text{max}},
\eeqna
where $\uvec_{\text{max}}$ is the normalized eigenvector corresponding to the maximal eigenvalue, and the last equality is obtained since $h_{\text{R}}^{+}(\Lmat)$ is not a function of $\dvec$.
By substituting \eqref{derivative_square_inv_K}  in \eqref{gradient_wmse_step1} we obtain
\beqna\label{gradient_wmse_step2}
\frac{\partial \lambda_{\text{max}}\left( {{h}_{\text{R}}^{+}} \Kmat^{-2} {{h}_{\text{R}}^{+}} \right)}{\partial d_n}
=-2\uvec_{\text{max}}^{T} {{h}_{\text{R}}^{+}}\Kmat^{-2}\frac{\partial\Kmat}{\partial d_n}\Kmat^{-1}{{h}_{\text{R}}^{+}} \uvec_{\text{max}}.
\eeqna
By substituting \eqref{derivative_K} in \eqref{gradient_wmse_step2} we obtain
\beqna\label{gradient_wmse_step3}
\frac{\partial \lambda_{\text{max}}\left( {{h}_{\text{R}}^{+}} \Kmat^{-2} {{h}_{\text{R}}^{+}} \right)}{\partial d_n}
=-2\uvec_{\text{max}}^{T} {{h}_{\text{R}}^{+}}\Kmat^{-2}{{h}_{\text{M}}} \hspace{1.8cm}\nonumber\\ \times\Big( \Emat_n \Rmat^{-1} \Dmat
 + \Dmat \Rmat^{-1} \Emat_n \Big) {{h}_{\text{M}}}\Kmat^{-1}{{h}_{\text{R}}^{+}} \uvec_{\text{max}}.\hspace{0cm}
\eeqna
Since the trace of a scalar is the scalar itself, the property ${\text{tr}}(\Amat\Bmat)={\text{tr}}(\Bmat\Amat)$ and the property ${\text{tr}}(\Amat)={\text{tr}}(\Amat^T)$ we obtain
\beqna\label{gradient_wmse_step4}
\frac{\partial \lambda_{\text{max}}\left( {{h}_{\text{R}}^{+}} \Kmat^{-2} {{h}_{\text{R}}^{+}} \right)}{\partial d_n}
=-2{\text{tr}}\Big( \Emat_n \Rmat^{-1} \Dmat {{h}_{\text{M}}}\Kmat^{-1}\hspace{1.2cm}\nonumber\\ \times\big(\Kmat^{-1}{{h}_{\text{R}}^{+}} \uvec_{\text{max}}\uvec_{\text{max}}^{T} {{h}_{\text{R}}^{+}}+{{h}_{\text{R}}^{+}} \uvec_{\text{max}}\uvec_{\text{max}}^{T} {{h}_{\text{R}}^{+}}\Kmat^{-1}\big)\Kmat^{-1}{{h}_{\text{M}}}\Big).
\eeqna
By noting that ${\text{tr}}(\Emat_n\Amat)=\Amat_{n,n}$ we obtain
\beqna\label{gradient_wmse_step5}
\nabla \lambda_{\text{max}}\left( {{h}_{\text{R}}^{+}} \Kmat^{-2} {{h}_{\text{R}}^{+}} \right)
=-2{\text{diag}}\Big(\Rmat^{-1} \Dmat {{h}_{\text{M}}}\Kmat^{-1}\hspace{1.2cm}\nonumber\\\times\big(\Kmat^{-1}{{h}_{\text{R}}^{+}} \uvec_{\text{max}}\uvec_{\text{max}}^{T} {{h}_{\text{R}}^{+}}+{{h}_{\text{R}}^{+}} \uvec_{\text{max}}\uvec_{\text{max}}^{T} {{h}_{\text{R}}^{+}}\Kmat^{-1}\big)\Kmat^{-1}{{h}_{\text{M}}}\Big).\hspace{0cm}
\eeqna
Substituting \eqref{gradient_bcrb} into \eqref{general_gradient_form} to obtain $\nabla \text{bCRB}(\dvec)$, and then substituting it and \eqref{gradient_wmse_step5} into \eqref{gradient_wcmse_1step}, yields \eqref{general_gradient_form} with $\Qmat$ from \eqref{gradient_wcmse}.
\vspace{-0.2cm}
\subsection{Derivation of the Gradient of \eqref{wc_mse_bayesian}}\label{gradient_wcbmse_derivation}
\vspace{-0.05cm}
In the following we derive the gradient of \eqref{wc_mse_bayesian}. First, we note that 
$\lambda_{\text{max}}(\Kmat^{-1})=\lambda_{\text{min}}^{-1}(\Kmat).$
The eigenvalue $\lambda_{\text{min}}\left( \Kmat \right)$ depends on $\dvec$ through $\Kmat$. The gradient of an eigenvalue \ac{wrt} a parameter is given by (see, e.g. Eq. (67) in \cite{matrix_cookbook})
\beqna\label{gradient_wcbmse_derivation_step1}
\frac{\partial \lambda_{\text{min}}\left( \Kmat\right)}{\partial d_n} = \uvec_{\text{min}}^{T} \frac{\partial \Kmat}{\partial d_n} \uvec_{\text{min}},
\eeqna
where $\uvec_{\text{min}}$ is the normalized eigenvector corresponding to $\lambda_{\text{min}}\left( \Kmat \right)$.
By substituting \eqref{derivative_K} in \eqref{gradient_wcbmse_derivation_step1} we obtain
\beqna\label{gradient_wcbmse_derivation_step2}
\frac{\partial \lambda_{\text{min}}^{-1}(\Kmat)}{\partial d_n}
\hspace{6.6cm}\nonumber\\  =-\lambda_{\text{min}}^{-2}(\Kmat)\uvec_{\text{min}}^{T} {{h}_{\text{M}}} \Big( \Emat_n \Rmat^{-1} \Dmat + \Dmat \Rmat^{-1} \Emat_n \Big) {{h}_{\text{M}}} \uvec_{\text{min}}
\hspace{0.5cm}\nonumber\\
=-\lambda_{\text{min}}^{-2}(\Kmat) \text{tr}\big(( \Emat_n \Rmat^{-1} \Dmat + \Dmat \Rmat^{-1} \Emat_n) {{h}_{\text{M}}} \uvec_{\text{min}}\uvec_{\text{min}}^{T} {{h}_{\text{M}}}\big).
\eeqna
By using the 
trace operator properties, we obtain
\beqna\label{gradient_wcbmse_derivation_step4}
\frac{\partial \lambda_{\text{min}}^{-1}(\Kmat)}{\partial d_n}
=-2\lambda_{\text{min}}^{-2}(\Kmat)   {\text{tr}}\Big( \Emat_n \Rmat^{-1} \Dmat {{h}_{\text{M}}} \uvec_{\text{min}}\uvec_{\text{min}}^{T} {{h}_{\text{M}}}\Big).
\eeqna
Noting that ${\text{tr}}(\Emat_n\Amat)=\Amat_{n,n}$, yields \eqref{general_gradient_form} with $\Qmat$ given by \eqref{gradient_wc_bmse}.

\vspace{-0.35cm}
\bibliographystyle{IEEEtran}
\bibliography{IEEEabrv,refs} 	
\vspace{0.5cm}
{\Large{\centering Supplemental Material for the Paper: \\ Efficient Sampling Allocation Strategies for General Graph-Filter-Based Signal Recovery}}

\vspace{0.5cm}
This document contains supplemental material for the paper \cite{Lital_25}, using the same notation. 
The following appendices are ordered to reflect the logical derivation path: In Subsection~\ref{ss;bandlimited_estimation} we derive the estimation method for strictly bandlimited graph signals. 
In Subsection~\ref{appendix_relation_bandlimited} we prove Claim \ref{claim_A_E}, which describes the corresponding \ac{mse} expression and the four associated cost functions.
In Subsection \ref{app;limit_cost_func}, we analyze the asymptotic behavior of the \ac{gfr}-\ac{ml} estimator and the proposed cost functions for $\mu\to\infty$. Subsection \ref{appendix_theorem_bandlimited} proves supporting matrix identities. Finally, Subsections \ref{sec_additive}, 
\ref{sec_submodularity_bmse}, and \ref{convexity_proof} discuss the additivity, submodularity and convexity of the cost functions.
The following table presents the assumptions of each section in the following. 
\vspace{-0.25cm}
\begin{table}[h]
\centering
\caption{Assumptions Summary Across Sections}
\label{tab:mu_assumptions}
\begin{tabular}{|c|c|c|c|c|}
\hline
\textbf{Sec.} & $\mu$ & $\xvec_0$ & $h_R(\Lmat)$ & $h_M(\Lmat)$ \\
\hline
S.I & $\to \infty$ & $ \in \ker(h_R(\Lmat))$ & $\Vmat^\top \operatorname{diag}(\onevec_{\mathcal{V} \setminus \mathcal{R}})\Vmat$ & $\Imat$ \\
\hline
S.II & $\to \infty$ & $ \in \ker(h_R(\Lmat))$ & $\Vmat^\top \operatorname{diag}(\onevec_{\mathcal{V} \setminus \mathcal{R}})\Vmat$ & $\Imat$ \\
\hline
S.III & $\to \infty$ & general & $h_R(\lambda_i) = 0,~ \forall i \in \mathcal{R}$ & general \\
\hline
S.IV & $ \to \infty$ & general & general & general \\
\hline
\end{tabular}
\end{table}


\vspace{-0.5cm}
\section{Bandlimited Graph Signal Recovery}\label{ss;bandlimited_estimation}
In this section, we develop the  GFR-ML estimator from \eqref{estimator}-\eqref{estimator_matrix} in \cite{Lital_25} for the special case of bandlimited graph signals, as defined in Subsection~\ref{special_case_bandlited}.The following derivations show that our general GFR-ML estimator reduces to the known GLS form under strict bandlimitedness (see, e.g. \cite{ortega2022introduction}).

In the strictly bandlimited setting (see Subsection \ref{special_case_bandlited}), it is assumed that the signal \( \xvec \) is confined to a subset of graph frequencies \( \mathcal{R} \subset \mathcal{V} \)  so that \( \xvec = \Umat_{\mathcal{R}} \tilde{\xvec} \), for some coefficient vector \( \tilde{\xvec} \in \mathbb{R}^{|\mathcal{R}|} \). This setting corresponds to the regularizer\footnote{In fact, the derivation and the results in this section are the same for any filter that satisfy $h_{\text{R}}^{+}(\lambda_i)=0, \forall i\in\mathcal{R}$.}
\be
\label{regularization_bandli}
h_{\text{R}}^{+}(\Lmat) = \Vmat\, \text{diag}(\onevec_{\mathcal{V} \setminus \mathcal{R}})\, \Vmat^{-1},
\ee
where \( \onevec_{\mathcal{V} \setminus \mathcal{R}} \in \mathbb{R}^N \) is the indicator vector for frequencies outside \( \mathcal{R} \), and \( \Vmat \) is the graph Fourier basis. 
This regularizer is designed as a projection onto the complement of the bandlimited subspace. 
The signal prior \( \xvec_0 \) is then assumed to satisfy \( \xvec_0 \in \text{ker}(h_{\text{R}}^{+}(\Lmat)) \). As a result,
$
h_{\text{R}}^{+}(\Lmat)\xvec_0 = \zerovec$.
In addition, to represent strictly bandlimited graph signals using the generalized regularization in \eqref{smothness_full_theta} with the prior \( \xvec_0 \) and the regularizer from \eqref{regularization_bandli}, it is required that \( \varepsilon = 0 \), corresponding to the limit \( \mu \to \infty \). 
Thus, under these assumptions, our general 
 GFR-ML estimator from \eqref{estimator} is reduced to
\be
\hat{\xvec} = \Kmat^{-1}(\dvec)\, h_{\text{M}}^T(\Lmat)\Dmat\Rmat^{-1}\Dmat\yvec.
\ee

Since $\tilde{\xvec}=\Vmat^{-1} \xvec$, we can write the ML estimator of the GFT of $\xvec$ as 
\be
\label{hat_tilde_x1}
\hat{\tilde{\xvec}}= \Vmat^{-1} \hat{\xvec} =\Vmat^{-1}\Kmat^{-1}(\dvec)\, h_{\text{M}}^T(\Lmat)\Dmat\Rmat^{-1}\Dmat\yvec.
\ee
Using the selection operator $\Pmat_{\mathcal{R}} \in \mathbb{R}^{|\mathcal{R}| \times N}$ that extracts the entries of $\tilde{\xvec}$ corresponding to the non-zero frequencies in $\mathcal{R}$, \eqref{hat_tilde_x1} can be written as
\beqna
\label{estimator_app2}
\hat{\tilde{\xvec}}_{\mathcal{R}} = \Pmat_{\mathcal{R}} \Vmat^{-1} \Kmat^{-1}(\dvec)\, h_{\text{M}}^T(\Lmat)\Dmat\Rmat^{-1}\Dmat\yvec
\nonumber\\=\Umat_{\mathcal{R}}^T \Kmat^{-1}(\dvec)\, h_{\text{M}}^T(\Lmat)\Dmat\Rmat^{-1}\Dmat\yvec,
\eeqna
where we use the fact that
 \( \Umat_{\mathcal{R}} = \Vmat \Pmat_{\mathcal{R}}^T \), and 
$
\Pmat_{\mathcal{R}} \Vmat^{-1} = \Umat_{\mathcal{R}}^T$.

Moreover, in the strictly bandlimited case, the signal lies in the subspace spanned by \( \Umat_{\mathcal{R}} \), and thus the effect of the regularization term in \( \Kmat^{\dagger}(\dvec) \) vanishes. That is,   \( \Kmat^{\dagger}(\dvec) \) from \eqref{estimator_matrix} for $\mu\to\infty$ becomes (see Theorem \ref{asymptotic_k} in \cite{Lital_25})
\begin{equation}
\Vmat^{-1}\Kmat^{\dagger}(\dvec)\Vmat= \begin{bmatrix}
(\Umat_{\mathcal{R}}^{T}\Kmat_{\text{M}}\Umat_{\mathcal{R}})^{-1} & \zerovec \\ \zerovec & \zerovec
\end{bmatrix}
,\label{KKK}
\end{equation}
where
\be
\label{K_m_def}
\Kmat_{\text{M}} \define {h}_{\text{M}}(\Lmat)\Dmat\Rmat^{-1}\Dmat{h}_{\text{M}}(\Lmat),
\ee
and $\Umat_{\mathcal{R}}^T \Kmat_{\text{M}} \Umat_{\mathcal{R}}$ is a non-singular matrix. 

To proceed, we multiply both sides of \eqref{KKK} by $\Pmat_{\mathcal{R}}$ from the left and by $\Vmat^{-1}$ from the right to obtain
\be
\label{relation1}
\Umat_{\mathcal{R}}^{T} \Kmat^{\dagger}(\dvec) = (\Umat_{\mathcal{R}}^{T}\Kmat_{\text{M}}\Umat_{\mathcal{R}})^{-1} \Umat_{\mathcal{R}}^T,
\ee
where we used the property that the columns of $\Vmat^{-1}\Kmat^{\dagger}(\dvec)\Vmat$ correspond to the complement of the bandlimited subspace are zero.
Substituting \eqref{relation1} into the estimator \eqref{estimator_app2} results in
\be
\label{hat_tilde_x2}
\hat{\tilde{\xvec}}_{\mathcal{R}} 
= \left( \Umat_{\mathcal{R}}^T \Kmat_{\text{M}} \Umat_{\mathcal{R}} \right)^{-1} \Umat_{\mathcal{R}}^T h_{\text{M}}^T(\Lmat)\Dmat\Rmat^{-1}\Dmat\yvec.
\ee

Now, we define the effective measurement matrix as
\be
\Hmat_{\mathcal{R}} \triangleq h_{\text{M}}(\Lmat) \Umat_{\mathcal{R}} \in \mathbb{R}^{N \times |\mathcal{R}|}.
\ee
Using this definition and \eqref{KKK}, we obtain
\[
\Umat_{\mathcal{R}}^T \Kmat_{\text{M}} \Umat_{\mathcal{R}} = \Hmat_{\mathcal{R}}^T \Dmat \Rmat^{-1} \Dmat \Hmat_{\mathcal{R}}, \quad
\Umat_{\mathcal{R}}^T h_{\text{M}}^T(\Lmat) = \Hmat_{\mathcal{R}}^T,
\]
which implies that \eqref{hat_tilde_x2} can be rewritten as
\be
\label{hat_tilde_x3}
\hat{\tilde{\xvec}}_{\mathcal{R}} 
= \left( \Hmat_{\mathcal{R}}^T \Dmat \Rmat^{-1} \Dmat \Hmat_{\mathcal{R}} \right)^{-1}
\Hmat_{\mathcal{R}}^T \Dmat \Rmat^{-1} \yvec.
\ee

The corresponding vertex-domain estimator \( \hat{\xvec} \in \mathbb{R}^N \) is obtained by projecting the estimated GFT coefficients in \eqref{hat_tilde_x3} back to the vertex domain using the bandlimited eigenbasis \( \Umat_{\mathcal{R}} \). Thus, we have
\be
\hat{\xvec} = \Umat_{\mathcal{R}} \hat{\tilde{\xvec}}_{\mathcal{R}}
 = \Umat_{\mathcal{R}} \left( \Hmat_{\mathcal{R}}^T \Dmat^T \Rmat^{-1} \Dmat \Hmat_{\mathcal{R}} \right)^{-1}
\Hmat_{\mathcal{R}}^T \Dmat^T \Rmat^{-1} \yvec.
\ee
For $h_{\text{M}}=\Imat$, this estimator is the classical \ac{gls} estimator for a strictly bandlimited graph signal under additive Gaussian noise. 

The \ac{mse} matrix of this estimator  is given by
\beqna
\label{MSE_for_bandli}
\text{MSE}(\hat{\xvec}) &=& \mathbb{E}\left[ (\hat{\xvec} - \xvec)(\hat{\xvec} - \xvec)^T;\xvec \right]
\nonumber\\&= &
\Umat_{\mathcal{R}} \left( \Hmat_{\mathcal{R}}^T \Dmat^T \Rmat^{-1} \Dmat \Hmat_{\mathcal{R}} \right)^{-1} \Umat_{\mathcal{R}}^T.
\eeqna
This expression is independent of \( \xvec \), making \ac{mse} minimization feasible, and is widely used in works assuming strict bandlimitedness. However, in the general case, the \ac{mse} matrix is a function of $\xvec$.

\section{Proof of Claim \ref{claim_A_E}}\label{appendix_relation_bandlimited}
To prove Claim~\ref{claim_A_E}, we consider the strictly bandlimited case described in Subsection~\ref{special_case_bandlited}, with \( \mu \to \infty \) and \( h_{\text{M}}(\Lmat) = \Imat \), and show that the cost functions in \eqref{CRB}, \eqref{mse_bayesian}, and \eqref{wc_mse_bayesian} are reduced to the A- and E-design criteria in \eqref{a_design} and \eqref{e_design}.

All relevant cost functions involve the matrix \( \Kmat_{\text{M}} \) from \eqref{K_m_def}, so we begin by rewriting this term. Let \( \Dmat \) be the sampling operator indicating the set \( \mathcal{S} \). Substituting \( h_{\text{M}}(\Lmat) = \Imat \) into the definition of \( \Kmat_{\text{M}} \) in \eqref{K_m_def} gives
\begin{equation}
\label{k_m_bandlimited}
\Umat_{\mathcal{R}}^{T}\Kmat_{\text{M}}\Umat_{\mathcal{R}} = \Umat_{\mathcal{R}}^{T}\Dmat \Rmat^{-1} \Dmat \Umat_{\mathcal{R}}= \Vmat_{\mathcal{S},\mathcal{R}}^T\Rmat_{\mathcal{S}}^{-1}\Vmat_{\mathcal{S},\mathcal{R}},
\end{equation}
where the second equality is obtained by substituting the definition of $\Umat_{\mathcal{R}}$ and using the masking property of  $\Dmat$.

Substituting \eqref{special_case}, \eqref{k_m_bandlimited}, and the asymptotic expression in \eqref{limit_K} from Theorem~\ref{asymptotic_k} into \eqref{CRB}, we obtain
\begin{equation}
\label{CRB_bandlimited}
\text{bCRB}(\dvec) = \text{tr} \left( ( \Vmat_{\mathcal{S},\mathcal{R}}^T\Rmat_{\mathcal{S},\mathcal{S}}^{-1}\Vmat_{\mathcal{S},\mathcal{R}} )^{-1} \right).
\end{equation}
Similarly, substituting the same expressions into \eqref{mse_bayesian} and \eqref{wc_mse_bayesian}  results in
\begin{equation}
\label{mse_bayesian_bandlimited}
\text{BMSE}(\dvec) = \text{tr} \left( ( \Vmat_{\mathcal{S},\mathcal{R}}^T\Rmat_{\mathcal{S},\mathcal{S}}^{-1}\Vmat_{\mathcal{S},\mathcal{R}} )^{-1} \right)
\end{equation}
\begin{equation}
\label{wc_mse_bayesian_bandlimited}
\text{BMSE}_{\text{WC}}(\dvec) = \lambda_{\text{max}} \left( ( \Vmat_{\mathcal{S},\mathcal{R}}^T\Rmat_{\mathcal{S},\mathcal{S}}^{-1}\Vmat_{\mathcal{S},\mathcal{R}} )^{-1} \right).
\end{equation}
From \eqref{CRB_bandlimited} and \eqref{mse_bayesian_bandlimited}–\eqref{wc_mse_bayesian_bandlimited}, we conclude the proof of Claim~\ref{claim_A_E}. Therefore, for the special case of estimating strictly bandlimited graph signals defined in Subsection \ref{special_case_bandlited}, and with \( h_{\text{M}}(\Lmat) = \Imat \), the \ac{bcrb} and \ac{bmse} reduce to the A-design in \eqref{a_design}, while the \ac{wcbmse} reduces to the E-design in \eqref{e_design}.

\vspace{-0.25cm}
\section{The Proposed Cost Functions for $\mu\to\infty$ }\label{app;limit_cost_func}
In the following, we analyze the effect of letting $\mu \to \infty$ on the cost functions from \eqref{CRB}, 
\eqref{worst_mse}, \eqref{mse_bayesian}, and \eqref{wc_mse_bayesian}.  The practical implications of the asymptotic analysis lie in understanding how tuning \( \mu \) to be large affects estimation and sampling within the image and kernel subspaces of \( h_{\text{R}}^+(\Lmat) \). 
 We begin by deriving the general estimator when $\mu \to \infty$.

\begin{claim}
For $h_{\text{R}}^+(\lambda_i)=0, ~\forall i\in\mathcal{R}$, when $\mu\to\infty$ 
the \ac{ml} estimator of the GFT of \( \xvec \) from \eqref{estimator}–\eqref{estimator_matrix} can be written as
\beqna\label{estimator_bandlimit}
\lim_{\mu \to \infty}
\hat{\tilde{\xvec}}=\Vmat^T\Kmat^{-1}(\dvec)({h}_{\text{M}}^T(\Lmat)\Dmat\Rmat^{-1}\Dmat\yvec+\mu{h}_{\text{R}}^{+}(\Lmat)\xvec_0)\nonumber\\=\begin{bmatrix}
    (\Umat_{\mathcal{R}}^{T}\Kmat_{\text{M}}\Umat_{\mathcal{R}})^{-1}\Umat_{\mathcal{R}}^T{h}_{\text{M}}^T(\Lmat)\Dmat\Rmat^{-1}\Dmat\yvec\\\zerovec
\end{bmatrix}\hspace{0.5cm}\nonumber\\+\begin{bmatrix}
 -\big(\Umat_{\mathcal{R}}^{T}\Kmat_{\text{M}}\Umat_{\mathcal{R}}\big)^{-1}\Umat_{\mathcal{R}}^{T}\Kmat_{\text{M}}\Umat_{\mathcal{V}\setminus\mathcal{R}}[\tilde{\xvec}_0]_{\mathcal{V}\setminus\mathcal{R}}\\
[\tilde{\xvec}_0]_{\mathcal{V}\setminus\mathcal{R}}
\end{bmatrix},
\eeqna
where $\Kmat_{\text{M}}$ is defined in \eqref{K_m_def}.
\end{claim}

\begin{IEEEproof}
Since \( \tilde{\xvec} = \Vmat^{-1} \xvec \),  
we can write the \ac{ml} estimator of the \ac{gft} of \( \xvec \) from \eqref{estimator} as
\begin{equation}
\label{hat_not_full_rank_x}
\hat{\tilde{\xvec}} = \Vmat^{-1} \hat{\xvec} = \Vmat^T \Kmat^{-1}(\dvec) \left( {h}_{\text{M}}^T(\Lmat)\Dmat\Rmat^{-1}\Dmat\yvec + \mu {h}_{\text{R}}^{+}(\Lmat)\xvec_0 \right).
\end{equation}

By taking the limit \( \mu \to \infty \) in \eqref{hat_not_full_rank_x} and substituting 
 \eqref{limit_K} and \eqref{limit_kh} from Appendix \ref{appendix_theorem_bandlimited}, the \ac{ml} estimator of the \ac{gft} of \( \xvec \) simplifies to \eqref{estimator_bandlimit}. 
\end{IEEEproof}

In this case, \( \hat{\tilde{\xvec}}_{\mathcal{R}} \) depends on the measurements and the prior, while \( \hat{\tilde{\xvec}}_{\mathcal{V} \setminus \mathcal{R}} \) depends only on \( [\tilde{\xvec}_0]_{\mathcal{V} \setminus \mathcal{R}} \), as appears in \eqref{Extreme_case_large_mu}. 

Notably, when \( h_{\text{R}}^+(\Lmat) \) is full column rank, which corresponds mathematically to \( \mathcal{V} \setminus \mathcal{R} = \mathcal{V} \) in \eqref{estimator_bandlimit}, \eqref{estimator_bandlimit} yields a degenerate case in which the estimation depends solely on the prior. This implies that setting \( \mu \) too large effectively eliminates the influence of the measurements on the estimator.

Having derived the \ac{gfr}-\ac{ml} estimator in \eqref{estimator_bandlimit}, we now turn to analyzing how the cost functions in \eqref{CRB}, \eqref{worst_mse}, \eqref{mse_bayesian}, and \eqref{wc_mse_bayesian} behave in this case.
By substituting \eqref{limit_K} into \eqref{CRB}, we obtain
\begin{equation}
\label{CRB_bandlimited2}
    \text{bCRB}(\dvec) = \text{tr} \Big( \big( \Umat_{\mathcal{R}}^{T} \Kmat_{\text{M}} \Umat_{\mathcal{R}} \big)^{-1} \Big).
\end{equation}
Substituting \eqref{limit_hkkh} into \eqref{worst_mse} without constant terms yields
\beqna\label{worst_mse_bandlimited}
\text{MSE}_{WC}(\dvec) 
 ={\text{bCRB}}({\dvec})\hspace{4.25cm}\nonumber\\
+ \lambda_{\text{max}}\Big(\Umat_{\mathcal{V}\setminus\mathcal{R}}^T\Kmat_{\text{M}}\Umat_{\mathcal{R}}\big(\Umat_{\mathcal{R}}^{T}\Kmat_{\text{M}}\Umat_{\mathcal{R}}\big)^{-2}\Umat_{\mathcal{R}}^{T}\Kmat_{\text{M}}\Umat_{\mathcal{V}\setminus\mathcal{R}}\Big).
\eeqna
By substituting \eqref{limit_K} into \eqref{mse_bayesian}, we obtain
\beqna\label{mse_bayesian_bandlimited2}
{\text{BMSE}}({\dvec}) 
= {\text{tr}}((\Umat_{\mathcal{R}}^{T}\Kmat_{\text{M}}\Umat_{\mathcal{R}})^{-1}).
\eeqna
Similarly, substituting \eqref{limit_K} into \eqref{wc_mse_bayesian} gives
\beqna\label{wc_mse_bayesian_bandlimited2} {\text{BMSE}_{WC}}(\dvec) = \lambda_{\text{max}}\Big((\Umat_{\mathcal{R}}^{T}\Kmat_{\text{M}}\Umat_{\mathcal{R}})^{-1}\Big). \eeqna

Notably, when \( h_{\text{R}}^+(\Lmat) \) is full column rank, which corresponds mathematically to \( \mathcal{V} \setminus \mathcal{R} = \mathcal{V} \) in \eqref{estimator_bandlimit}, \eqref{estimator_bandlimit} yields a degenerate case in which the estimator relies solely on the prior and becomes independent of the observed data. As a result,
the cost functions in \eqref{CRB_bandlimited2}-\eqref{wc_mse_bayesian_bandlimited2} are independent of the measurements, as expected, and there is no need for sampling design.

Finally, it can be seen that if $\Kmat_{\text{M}}$ satisfies \eqref{k_m_bandlimited}, the cost functions in \eqref{CRB_bandlimited2}, 
\eqref{mse_bayesian_bandlimited2}, and \eqref{wc_mse_bayesian_bandlimited2}
are reduced to 
\eqref{CRB_bandlimited}, \eqref{mse_bayesian_bandlimited},
and \eqref{wc_mse_bayesian_bandlimited}
from Section \ref{appendix_relation_bandlimited}. Thus, this section extends the analysis to a more general setting with an arbitrary prior \( \xvec_0 \), a general measurement graph filter \( h_{\text{M}}(\Lmat) \), and a more flexible regularization graph filter \( h_{\text{R}}(\Lmat) \).

\section{Asymptotic Behavior of $\Kmat(\dvec)$ for $\mu\to\infty$}\label{appendix_theorem_bandlimited}
In this appendix, we establish the asymptotic behavior of the \ac{gfr}-\ac{ml} estimator as the regularization parameter \( \mu \to \infty \). 
Understanding the asymptotic behavior of 
$\Kmat
(\dvec)$  as 
\( \mu \to \infty \) is important for characterizing the estimator's sensitivity to the regularization term. This analysis reveals  how the estimator and the proposed cost functions in \eqref{CRB}, \eqref{worst_mse}, \eqref{mse_bayesian}, and \eqref{wc_mse_bayesian} behave when the prior dominates, demonstrating mathematically that the estimator depends solely on the prior and, correspondingly, that the cost functions become independent of the data.
Furthermore, since bandlimited graph signal recovery corresponds to the special case where \( \mu \to \infty \) for specific choices of \( \xvec_0 \), \( h_{\text{R}} \), and \( h_{\text{M}} \), this asymptotic analysis also helps clarify the connection to classical bandlimited recovery. It is particularly useful for the proof of Claim~\ref{claim_A_E}, which establishes the relationship between the proposed cost functions and the A-/E-design methods.

To this end, we first analyze the inverse of a general affine block matrix in the following Lemma.
\begin{lemma}[Asymptotic inverse of an affine block matrix]\label{thm:asymptotic_mu}
Let
\(\Amat\in\mathbb{C}^{n\times n}\),
\(\Bmat\in\mathbb{C}^{n\times m}\), and
\(\Dmat,\Emat\in\mathbb{C}^{m\times m}\),  
where \(\Amat\) and \(\Emat\) are nonsingular.
In addition, assume that the block matrix
\begin{equation}
\Mmat(\mu)\triangleq
\begin{bmatrix}
\Amat & \Bmat\\[2pt]
\Bmat^T & \Dmat+\mu \Emat
\end{bmatrix}
\label{eq:def_Mmu}
\end{equation}
is a non-singular matrix. Then,
for a sufficiently large $\mu$, the following hold:
\vspace{-0.1cm}
    \begin{equation}
    \lim_{\mu \to \infty} \Mmat^{-1}(\mu) =
    \begin{bmatrix}
    \Amat^{-1} & \zerovec\\[2pt]
    \zerovec & \zerovec
    \end{bmatrix},\hspace{2cm}
    \label{eq:limit_inv_Mmu}
    \end{equation}
    \vspace{-0.2cm}
    \begin{equation}
    \lim_{\mu \to \infty} \Mmat^{-1}(\mu)\begin{bmatrix}
\zerovec & \zerovec\\[2pt]
\zerovec & \mu \Emat
\end{bmatrix} =
    \begin{bmatrix}
    \zerovec & -\Amat^{-1}\Bmat\\[2pt]
    \zerovec & \Imat
    \end{bmatrix}.
    \label{eq:limit_inv_MmuE}
    \end{equation}
\end{lemma}

\begin{IEEEproof}
According to the lemma for the inverse of matrices \cite[Eq. (400)]{matrix_cookbook}, the inverse of  \( \Mmat(\mu) \) can be written as
\beqna
\Mmat^{-1}(\mu)=\hspace{6cm}
\nonumber\\
\begin{bmatrix}
\Amat^{-1}+\Amat^{-1}\Bmat\,\Smat^{-1}(\mu)\Bmat^T\,\Amat^{-1} & -\Amat^{-1}\Bmat\,\Smat^{-1}(\mu)\\[4pt]
-\Smat^{-1}(\mu)\Bmat^T\,\Amat^{-1} & \Smat^{-1}(\mu)
\end{bmatrix},
\label{eq:Mmu_inv_block}
\eeqna
where 
\begin{equation}
\Smat(\mu)\triangleq \Dmat+\mu \Emat-\Bmat^T\Amat^{-1}\Bmat
\label{eq:Schur_def}
\end{equation}
is its associated Schur complement.
Since \( \Emat \) is a nonsingular matrix, 
 for sufficiently large \( \mu \) such that
 \[\mu^{-1}\lambda_{\rm{min}}(\Emat^{-1}(\Dmat-\Bmat^T\Amat^{-1})>-1,\] 
 we can write the inverse of $\Smat(\mu)$ as follows:
\begin{equation}
\Smat^{-1}(\mu)
   =\mu^{-1}\Emat^{-1}\!\bigl(\Imat+\mu^{-1}\Emat^{-1}(\Dmat-\Bmat^T\Amat^{-1}\Bmat)\bigr)^{-1}.
\label{eq:Schur_inv_scaled}
\end{equation}
Taking the limit of \( \mu \to \infty \) in \eqref{eq:Schur_inv_scaled}, one obtains
\begin{equation}
\lim_{\mu\to\infty}\mu \Emat\Smat^{-1}(\mu)
=\Imat.
\label{eq:mu_S_converge}
\end{equation}
Consequently the entries of 
\(\Smat^{-1}(\mu)\) converge to zero as \( \mu \to \infty \) with $\left| [\Smat^{-1}(\mu)]_{i,j} \right|=\mathcal{O}(\mu^{-1})$, $\forall i,j\in\mathcal{V}$.

Hence, 
applying the limit \( \mu \to \infty \) to \eqref{eq:Mmu_inv_block} and using \eqref{eq:mu_S_converge}, each term involving \( \Smat^{-1}(\mu) \) vanishes, yielding the result in \eqref{eq:limit_inv_Mmu}.
Moreover, using matrix multiplication rules, we have
\begin{equation}
\Mmat^{-1}(\mu)\begin{bmatrix}
\zerovec & \zerovec\\[2pt]
\zerovec & \mu \Emat
\end{bmatrix}=
\begin{bmatrix}
\zerovec & -\mu\Amat^{-1}\Bmat\,\Smat^{-1}(\mu)\Emat\\[4pt]
\zerovec & \mu\Smat^{-1}(\mu)\Emat
\end{bmatrix}.
\label{eq:MmuE_inv_block}
\end{equation}
Taking the limit \( \mu \to \infty \) and applying \eqref{eq:mu_S_converge}, we obtain the result in \eqref{eq:limit_inv_MmuE}.
\end{IEEEproof}

Now we use Lemma \ref{thm:asymptotic_mu} to develop the asymptotic behavior of $\Kmat(\dvec)$.
The proof is based on the block structure of the transformed system matrix.
\begin{theorem}
\label{asymptotic_k}
Let \( \Kmat(\dvec) \) be defined as in~\eqref{estimator_matrix}, where \( h_{\text{M}}(\Lmat) \) and \( h_{\text{R}}^{+}(\Lmat) \) are graph filters, and  \( h_{\text{R}}^{+}(\Lmat)\Umat_{\mathcal{R}} = \zerovec \). 
If \( \Umat_{\mathcal{R}}^{T} \Kmat_{\text{M}} \Umat_{\mathcal{R}} \) and \( \Umat_{\mathcal{V} \setminus \mathcal{R}}^{T} h_{\text{R}}^{+}(\Lmat) \Umat_{\mathcal{V} \setminus \mathcal{R}} \) are invertible, then the following asymptotic properties hold as \( \mu \to \infty \):
\beqna
\lim_{\mu \to \infty} \Vmat^{-1}\Kmat^{\dagger}(\dvec)\Vmat= \begin{bmatrix}
(\Umat_{\mathcal{R}}^{T}\Kmat_{\text{M}}\Umat_{\mathcal{R}})^{-1} & \zerovec \\ \zerovec & \zerovec
\end{bmatrix}
,\hspace{0.7cm}\label{limit_K}
\eeqna
\beqna
\lim_{\mu \to \infty} \mu\Vmat^{-1}\Kmat^{\dagger}(\dvec)  {h}_{\text{R}}^{+}(\Lmat)\Vmat \hspace{3.5cm}\nonumber\\= \begin{bmatrix}
\zerovec & -\big(\Umat_{\mathcal{R}}^{T}\Kmat_{\text{M}}\Umat_{\mathcal{R}}\big)^{-1}\Umat_{\mathcal{R}}^{T}\Kmat_{\text{M}}\Umat_{\mathcal{V}\setminus\mathcal{R}}\\[4pt]
\zerovec & \Imat
\end{bmatrix}
,\label{limit_kh}
\eeqna
\begin{equation}
\lim_{\mu \to \infty} \mu^2\Vmat^{-1}{h}_{\text{R}}^{+}(\Lmat)\Kmat^{\dagger}(\dvec)\Kmat^{\dagger}(\dvec)  {h}_{\text{R}}^{+}(\Lmat)\Vmat \nonumber\hspace{1.5cm}
\end{equation}
\begin{equation}
    = \begin{bmatrix}
\zerovec & \zerovec\\[4pt]
\zerovec & \Umat_{\mathcal{V}\setminus\mathcal{R}}^T\Kmat_{\text{M}}\Umat_{\mathcal{R}}\big(\Umat_{\mathcal{R}}^{T}\Kmat_{\text{M}}\Umat_{\mathcal{R}}\big)^{-2}\Umat_{\mathcal{R}}^{T}\Kmat_{\text{M}}\Umat_{\mathcal{V}\setminus\mathcal{R}}+\Imat
\end{bmatrix}
.\label{limit_hkkh}
\end{equation}
\end{theorem}

\begin{IEEEproof}
  By 
  using the definition from \eqref{estimator_matrix} and the property that \( h_{\text{R}}^{+}(\Lmat)\Umat_{\mathcal{R}} = \zerovec \), one can verify that \( \Vmat^{-1}\Kmat(\dvec)\Vmat \) is a block matrix of the form given in \eqref{eq:def_Mmu}, where
  \vspace{-0.1cm}
\beqna\label{block_A}
\Amat=\Umat_{\mathcal{R}}^{T}\Kmat_{\text{M}}\Umat_{\mathcal{R}},~\Bmat=\Umat_{\mathcal{R}}^{-1}\Kmat_{\text{M}}\Umat_{\mathcal{V}\setminus\mathcal{R}},\hspace{1cm} \nonumber\\ \Dmat=\Umat_{\mathcal{V}\setminus\mathcal{R}}^{T}\Kmat_{\text{M}}\Umat_{\mathcal{V}\setminus\mathcal{R}},~
\Emat=\Umat_{\mathcal{V}\setminus\mathcal{R}}^{T}h_{\text{R}}^{+}(\Lmat)\Umat_{\mathcal{V}\setminus\mathcal{R}}.
\eeqna
Since both \( \Amat \)  and \( \Emat \) defined in \eqref{block_A}
are invertible, \eqref{eq:limit_inv_Mmu} from Theorem~\ref{thm:asymptotic_mu} directly implies the result in \eqref{limit_K}.

 Consider the expression
   \vspace{-0.1cm}
\beqna
\Vmat^{-1}\mu\Kmat^{\dagger}(\dvec)  {h}_{\text{R}}^{+}(\Lmat)\Vmat=\Vmat^{-1}\mu\Kmat^{\dagger}(\dvec)\Vmat  \Vmat^{-1}{h}_{\text{R}}^{+}(\Lmat)\Vmat.
\eeqna
Note that by using the  assumption  that \( h_{\text{R}}^{+}(\Lmat)\Umat_{\mathcal{R}} = \zerovec \) and substituting the notations in  \eqref{block_A}, we have
  \vspace{-0.1cm}
\beqna
\Vmat^{-1}{h}_{\text{R}}^{+}(\Lmat)\Vmat=\begin{bmatrix}
\zerovec & \zerovec\\[2pt]
\zerovec & \mu \Emat
\end{bmatrix}.
\eeqna
Applying the limit \( \mu \to \infty \) and substituting \eqref{eq:limit_inv_MmuE} along with the definitions in \eqref{block_A}, we obtain the result in \eqref{limit_kh}.
Finally, since the limit in \eqref{limit_kh} exists and is finite, the limit of the matrix product
$
\lim_{\mu \to \infty} \Vmat^{-1} \mu^2 h_{\text{R}}^{+}(\Lmat) \Kmat^{\dagger}(\dvec) \Kmat^{\dagger}(\dvec) h_{\text{R}}^{+}(\Lmat) \Vmat
$
can be evaluated by multiplying the limit matrices obtained in \eqref{limit_kh}. This yields the result in \eqref{limit_hkkh}.
\end{IEEEproof}

\vspace{-0.25cm}
\section{Additive Property of  \( \Kmat(\dvec) \)}
\label{sec_additive}
In this section, we establish the additive (modular) structure of the estimator matrix 
$\Kmat(\dvec)$
 in \eqref{estimator_matrix}, a property that enables efficient updates in greedy algorithms and facilitates the submodularity proof of the \ac{bmse} objective. Proposition~\ref{theorem_modular_K} formalizes this property by showing that for diagonal noise covariance, each node contributes independently to 
$\Kmat(\dvec)$.
\begin{proposition}
\label{theorem_modular_K}
Consider a subset of nodes \(  \{a\}\in \mathcal{V}\setminus\mathcal{S} \), assuming a diagonal noise covariance matrix \( \Rmat \). The estimator matrix from \eqref{estimator_matrix} can be written as 
\beqna\label{modular_k}
\Kmat(\onevec_{\mathcal{S} \cup \{a\}}) = \Kmat(\onevec_{\mathcal{S}}) + \Kmat_{\text{M}}(\onevec_{\{a\}}),
\eeqna
where all the contribution of the nodes in $\{a\}$ is 
captured by \beqna\label{k_m_definition}
\Kmat_{\text{M}}(\onevec_{\{a\}})= \sum_{i\in\{a\}}{\Rmat^{-1}_{i,i}} [h_{\text{M}}(\Lmat)]_{\mathcal{V},i}[h_{\text{M}}(\Lmat)]_{\mathcal{V},i}^T. 
\eeqna 
\end{proposition}
\begin{IEEEproof}
For any \( \mathcal{S} \subseteq \mathcal{V} \) and diagonal $\Rmat$, it follows that
\begin{equation}
    \label{saparted_M}
    {h}_{\text{M}}(\Lmat)\Dmat\Rmat^{-1}\Dmat{h}_{\text{M}}(\Lmat)= \sum_{i \in \mathcal{S}}\frac{1}{\Rmat_{i,i}}[{h}_{\text{M}}(\Lmat)]_{\mathcal{V},i}[{h}_{\text{M}}(\Lmat)]_{\mathcal{V},i}^T.
\end{equation}
Substituting \eqref{saparted_M} for the set $\mathcal{S}\cup\{a\}$, with $\{a\}\notin\mathcal{S}$, into \eqref{estimator_matrix}, we obtain  \vspace{-0.1cm}
\beqna
    \Kmat(\onevec_{\mathcal{S}\cup\{a\}}) = \sum_{i \in \mathcal{S}\cup\{a\}}\frac{1}{\Rmat_{i,i}}[{h}_{\text{M}}(\Lmat)]_{\mathcal{V},i}[{h}_{\text{M}}(\Lmat)]_{\mathcal{V},i}^T+ \mu{h}_{\text{R}}^{+}(\Lmat)\nonumber\\
    = \Kmat(\onevec_{\mathcal{S}})+\sum_{i\in\{a\}}{\Rmat^{-1}_{i,i}} [h_{\text{M}}(\Lmat)]_{\mathcal{V},i}[h_{\text{M}}(\Lmat)]_{\mathcal{V},i}^T .\nonumber
\eeqna
\end{IEEEproof} 
\vspace{-0.3cm}
\section{Prove of Submodularity of the BMSE}\label{sec_submodularity_bmse}
In this section, we show that the negative \ac{bmse} cost function is submodular and monotonically increasing when $\Rmat$ is diagonal and $h_{\text{R}}^{+}(\Lmat)$ is positive definite. The submodularity proof follows a similar approach to that in \cite{Submodular39}.
The proof is based on the following definition of submodularity \cite{Submodular39}:
 \begin{definition}[submodularity]\label{submodularity}
A set function \( f : 2^{\mathcal{V}} \to \mathbb{R} \) is submodular if and only if the derived set functions 
\( f_a : 2^{\mathcal{V} \setminus \{a\}} \to \mathbb{R} \) defined by
\begin{equation}
    f_a(\mathcal{S}) = f(\mathcal{S} \cup \{a\}) - f(\mathcal{S})
    \label{eq:submodularity}
\end{equation}
are monotone decreasing, i.e., if for all subsets \(\mathcal{A},\mathcal{B} \subseteq \mathcal{V}\setminus\{a\} \) it holds that
$\mathcal{A} \subseteq \mathcal{B}$ imply that $f_a(\mathcal{B}) \leq f_a(\mathcal{A})$.
\end{definition}

Submodularity 
provide strong optimality guarantees for maximization problems. Therefore, we reformulate the original minimization in \eqref{eq:optimization_problem} as an equivalent maximization by negating the objective 
\begin{equation}
    \mathcal{S}^*=\arg\max_{\mathcal{S} \subseteq \mathcal{V},\, |\mathcal{S}| = q} -\text{BMSE}(\onevec_{\mathcal{S}}),
    \label{eq:set_opt}
\end{equation}
and prove these properties for the resulting negative objective. 


We begin with  defining the derived set function $f_a : 2^{\mathcal{V} \setminus \{a\}} \to \mathbb{R}$ for the negative \ac{bmse}  $\forall\{a\} \in \mathcal{V}$ as
\beqna
    f_a(\mathcal{S}) = -\operatorname{tr} \big(\Kmat^{-1}(\onevec_{\mathcal{S}\cup\{a\}}) \big) + \operatorname{tr}(\Kmat^{-1}(\onevec_{\mathcal{S}}) )\hspace{1cm}\nonumber\\
     = -\operatorname{tr} \Big(\big(\Kmat(\onevec_{\mathcal{S}}) +\Kmat_{\text{M}}(\onevec_{\{a\}})\big)^{-1}\Big)+ \operatorname{tr}(\Kmat^{-1}(\onevec_{\mathcal{S}}) ),
\eeqna
where second equality obtain by substituing \eqref{modular_k} from Proposition \ref{theorem_modular_K} and $\Kmat_{\text{M}}(\onevec_{\{a\}})$ defined in \eqref{K_m_def}. 
We define 
\[
   \tilde{\Kmat}(\theta) \define  \Kmat(\onevec_{\mathcal{S}_1}) + \theta \, ( \Kmat(\onevec_{\mathcal{S}_2}) -  \Kmat(\onevec_{\mathcal{S}_1}) ), \quad \theta \in [0,1],
\]
so that \( \tilde{\Kmat}(0) = \Kmat(\onevec_{\mathcal{S}_1}) \) and \( \tilde{\Kmat}(1) = \Kmat(\onevec_{\mathcal{S}_2}) \).
Now define
\[
    \hat{f}_a(\tilde{\Kmat}(\theta)) = -\operatorname{tr}\big( (\tilde{\Kmat}(\theta) + \Kmat_{\text{M}}(\onevec_{\{a\}}))^{-1} \big) + \operatorname{tr}\big(\tilde{\Kmat}^{-1}(\theta) \big).
\]
Note that \( \hat{f}_a(\tilde{\Kmat}(0)) = f_a(\mathcal{S}_1) \) and \(\hat{f}_a(\tilde{\Kmat}(1)) = f_a(\mathcal{S}_2) \).  
To show that $f_a(\mathcal{S})$ satisfies \eqref{eq:submodularity}, we show that the derivative of $\hat{f}_a(\tilde{\Kmat}(\theta))$ w.r.t. $\theta$ is non-positive. Differentiating yields
\begin{align}\label{der_1}
    \frac{d}{d\theta} \hat{f}_a(\tilde{\Kmat}(\theta))
    = \frac{d}{d\theta}\mathrm{tr}\left(\tilde{\Kmat}^{-1}(\theta) \right) \hspace{2cm}\nonumber\\- \frac{d}{d\theta} \, \mathrm{tr}\left( (\tilde{\Kmat}(\theta) + \Kmat_{\text{M}}(\onevec_{\{a\}}))^{-1} \right) 
       .
\end{align}
Using the matrix derivative formula \cite[Eq.~(40)]{matrix_cookbook}
\[
    \frac{d}{d\theta} \Xmat(\theta)^{-1} 
    = - \Xmat(\theta)^{-1} \frac{d \Xmat(\theta)}{d\theta} \Xmat(\theta)^{-1},
\]
and the cyclic property of the trace, we obtain
\begin{equation}\label{tr1_der}
    \frac{d}{d\theta}\mathrm{tr}(\tilde{\Kmat}^{-1}(\theta) )
     =-\mathrm{tr} (\tilde{\Kmat}^{-2}(\theta) \big( \Kmat(\onevec_{\mathcal{S}_2}) - \Kmat(\onevec_{\mathcal{S}_1}) \big)), 
     \end{equation}
\beqna \label{tr2_der}
    \frac{d}{d\theta} \, \mathrm{tr}\left( (\tilde{\Kmat}(\theta) + \Kmat_{\text{M}}(\onevec_{\{a\}}))^{-1} \right) \hspace{3.25cm}\nonumber\\
    = -\mathrm{tr} \left( (\tilde{\Kmat}(\theta) + \Kmat_{\text{M}}(\onevec_{\{a\}}))^{-2} \big( \Kmat(\onevec_{\mathcal{S}_2}) - \Kmat(\onevec_{\mathcal{S}_1}) \big) \right).
\eeqna
Substituting \eqref{tr1_der} and \eqref{tr2_der} into \eqref{der_1} gives
     \beqna\label{der2}
         \frac{d}{d\theta} \hat{f}_a(\tilde{\Kmat}(\theta))
    = \mathrm{tr} \Big((\tilde{\Kmat}(\theta) + \Kmat_{\text{M}}(\onevec_{\{a\}}))^{-2} - \big(\tilde{\Kmat}(\theta)^{-2}\big) \nonumber\\
    \times\big( \Kmat(\onevec_{\mathcal{S}_2}) - \Kmat(\onevec_{\mathcal{S}_1}) \big) \Big).
\eeqna
Based on Proposition \ref{theorem_modular_K}, for any \( \mathcal{S}_1 \subseteq \mathcal{S}_2 \subseteq \mathcal{V} \setminus \{a\} \), it holds that  
\( \Kmat(\onevec_{\mathcal{S}_2}) - \Kmat(\onevec_{\mathcal{S}_1}) \succeq 0 \).  
Similarly, we have  \( (\tilde{\Kmat}(\theta) + \Kmat_{\text{M}}(\onevec_{\{a\}}))^{-2} -\tilde{\Kmat}^{-2}(\theta) \preceq 0 \). 
Therefore, the matrix inside the trace on the right-hand side of \eqref{der2} is negative semidefinite, since it is the product of a positive semidefinite and a negative semidefinite matrix. Consequently, its trace is non-positive, which, when substituted into \eqref{der2}, yields the inequality
$
         \frac{d}{d\theta} \hat{f}_a(\tilde{\Kmat}(\theta))
         \leq 0.$
Integrating over $\theta\in[0,1]$ yields
\[
    \hat{f}_a(\tilde{\Kmat}(1)) - \hat{f}_a(\tilde{\Kmat}(0)) = \int_0^1 \frac{d}{d\theta} \hat{f}_a(\tilde{\Kmat}(\theta)) \, d\theta \leq 0,
\]
i.e., \( \hat{f}_a(\tilde{\Kmat}(1)) =f_a(\mathcal{S}_2) \leq  \hat{f}_a(\tilde{\Kmat}(0)) = f_a(\mathcal{S}_1) \).  
 Thus, \( f_a \) is monotone decreasing, and \( f \) is submodular by Definition~\ref{submodularity}.

\vspace{-0.25cm}
\section{Convexity of the \ac{bmse} and \ac{wcbmse} }\label{convexity_proof}
In this section, we prove that for diagonal $\Rmat$ and positive definite $h_{\text{R}}^{+}(\Lmat)$, the \ac{bmse} and \ac{wcbmse} are convex \ac{wrt} the (squared) selection weights
$w_i \triangleq d_i^2 \in [0,1]$.

For diagonal $\Rmat$, we can write the estimator matrix $\Kmat(\dvec)$ as a function of $\wvec$, where $[\wvec]_i=w_i$, as
\begin{equation}
\label{eq:K-affine}
\Kmat(\dvec) = \Kmat_{\text{affine}}(\wvec) \triangleq \sum_{i \in \mathcal{V}} \frac{w_i}{\Rmat_{i,i}} [{h}_{\text{M}}(\Lmat)]_{\mathcal{V},i} [{h}_{\text{M}}(\Lmat)]_{\mathcal{V},i}^T + \mu\,{h}_{\text{R}}^{+}(\Lmat).
\end{equation}
Here, $\Kmat_{\text{affine}}(\wvec)$ is affine in $\wvec$, and $\Kmat_{\text{affine}}(\wvec) \succ \zerovec$ on the feasible set (for $\mu>0$).

To prove the convexity of the \ac{bmse} and \ac{wcbmse}, we use the following Theorem \cite[Example~3.4]{Boyd_2004}:

\begin{theorem}[Matrix Fractional Function]\label{matrix_frac_theorem}
    The function 
    \[
        f : \mathbb{R}^N \times \mathbb{S}^N_{++} \to \mathbb{R}, 
        \quad f(\xvec,\Ymat) = \xvec^\top \Ymat^{-1} \xvec
    \]
    is convex on its domain 
    $\mathrm{dom}\,f = \mathbb{R}^N \times \mathbb{S}^N_{++}$.
\end{theorem}

\begin{proposition}
\label{prop:bmse-convex}
The function
$
\mathrm{tr}\big(\Kmat^{-1}_{\text{affine}}(\wvec)\big) $
is convex in nonnegative $\wvec$.
\end{proposition}

\begin{IEEEproof}
The trace operation can be expressed as a sum of matrix–fractional functions. Let $\{\mathbf e_i\}_{i=1}^N$ denote the standard basis vectors. Then $\mathrm{tr}\big(\Ymat^{-1}\big)
= \sum_{i=1}^N  \mathbf e_i^{\top}\Ymat^{-1} \mathbf e_i$,
where each term $\mathbf e_i^{\top}\Ymat^{-1} \mathbf e_i$ is a matrix–fractional function and is convex in $\Ymat\succ 0$ (Theorem~\ref{matrix_frac_theorem}). Therefore, $\mathrm{tr}(\Ymat^{-1})$ is convex on $\mathbb{S}_{++}^N$. Since $\Kmat_{\text{affine}}(\wvec)$ is affine in $\wvec$, and the composition of a convex function with an affine mapping preserves convexity, it follows that $\mathrm{tr}\big(\Kmat^{-1}_{\text{affine}}(\wvec)\big)$ is convex in nonnegative $\wvec$.
\end{IEEEproof}

\begin{proposition}
\label{prop:bmse-wc-convex}
The function $
\lambda_{\max}\big(\Kmat^{-1}_{\text{affine}}(\wvec)\big)$
is convex in nonnegative $\wvec$.
\end{proposition}

\begin{IEEEproof}
From the variational form $\lambda_{\max}(\Amat) = \max_{\|\xvec\|=1} \xvec^\top \Amat \xvec$, for any fixed $\xvec$ and $\Ymat \succ 0$, define $g_\xvec(\Ymat) \triangleq \xvec^\top \Ymat^{-1} \xvec$. By Theorem~\ref{matrix_frac_theorem}, $g_\xvec(\Ymat)$ is convex in $\Ymat$. Since $\Kmat_{\text{affine}}(\wvec)$ is affine in $\wvec$, $g_\xvec(\Kmat_{\text{affine}}(\wvec))$ is convex in $\wvec$. Hence, $\lambda_{\max}(\Kmat^{-1}_{\text{affine}}(\wvec))\hspace{-0.1cm} =\hspace{-0.1cm} \max_{\|\xvec\|=1} g_\xvec(\Kmat_{\text{affine}}(\wvec))
$ 
is the pointwise supremum of convex functions, i.e. convex.
\end{IEEEproof}

\end{document}